\documentclass[11pt]{amsart}
\RequirePackage[left=1in,right=1in,top=1in,bottom=1in]{geometry}  % Redefine the page margins
\usepackage{amsmath,amssymb,amsthm,amsfonts,bm,latexsym,graphicx,float,caption,subcaption,yfonts,cite,mathtools,dsfont,mathabx}
\usepackage[utf8]{inputenc}
\usepackage[T1]{fontenc}
\graphicspath{{figures/}}% Directory in which figures are stored
\usepackage[font=small,format=plain,labelfont=bf,textfont={footnotesize,it},justification=justified,singlelinecheck=true]{caption}%figure captions
\usepackage[table,xcdraw]{xcolor}
\usepackage{multirow}
\usepackage{comment}
\usepackage{tikz}

\def\R{\mathbb{R}}

\def\rR{\mathcal{R}}
\def\I{\mathbb{I}}

\def\hF{\hat{F}}

\def\hJ{\hat{J}}

\def\e{\varepsilon}

\newcommand{\eps}{\varepsilon}
\newcommand\norm[1]{\left\lVert#1\right\rVert}

\DeclareMathOperator*{\argmin}{arg\,min}

\makeatletter
\@namedef{subjclassname@2020}{%
  \textup{2020} Mathematics Subject Classification}
\makeatother

\newtheorem{theorem}{Theorem}
\newtheorem{prop}[theorem]{Proposition}
\newtheorem{lemma}[theorem]{Lemma}
\newtheorem{remark}[theorem]{Remark}

\newtheorem{definition}[theorem]{Definition}

\numberwithin{theorem}{section}

\everymath{\displaystyle}

\author{Guillaume Bal}
\address{Departments of Statistics and Mathematics and CCAM, University of Chicago, Chicago, IL 60637}
\email{guillaumebal@uchicago.edu}

\author{Ruoming Gong}
\address{Department of Engineering Science and Applied Mathematics, Northwestern University, Evanston, IL 60201}
\email{ruoming.gong@northwestern.edu }

\author{Fatma Terzioglu}
\address{Department of Mathematics, North Carolina State University, Raleigh, NC 27695}
\email[Corresponding author]{fterzioglu@ncsu.edu}

\date{}
\title[An inversion algorithm for $P-$functions]{An inversion algorithm for $P-$functions with applications to Multi-energy CT}


\begin{document}
\maketitle

\begin{abstract}
Multi-energy computed tomography (ME-CT) is an x-ray transmission imaging technique that uses the energy dependence of x-ray photon attenuation to determine the elemental composition of an object of interest. Mathematically, forward ME-CT measurements are modeled by a nonlinear integral transform. In this paper, local conditions for global invertibility of the ME-CT transform are studied, and explicit stability estimates quantifying the error propagation from measurements to reconstructions are provided. Motivated from the inverse problem of image reconstruction in ME-CT, an iterative inversion algorithm for the so-called $P-$functions is proposed. Numerical simulations for ME-CT, in two and three materials settings with an equal number of energy measurements, confirm the theoretical predictions. 

\end{abstract}
%%----------------------------------------------
\section{Introduction}
%%----------------------------------------------
Multi-energy computed tomography (ME-CT) is a diagnostic imaging technique that uses x-rays for identifying material properties of an examined object in a non-invasive manner. While standard CT is based on the simplifying assumption of mono-energetic radiation, ME-CT exploits the fact that the attenuation of x-ray photons depends on the energy of the x-ray photon in addition to the materials present in the imaged object \cite{AlvarezMacovski, Lionheart, Katsura, McCollough, Park}.

ME-CT employs several energy measurements acquired from either an energy integrating detector that uses different x-ray source energy spectra or a photon counting detector that can register photons in multiple energy windows \cite{Schlomka, Taguchi, Willemink}. Energy dependence of photon attenuation can be utilized to distinguish between different materials in an imaged object based on their density or atomic numbers \cite{AlvarezMacovski}. As a result, ME-CT provides quantitative information about the material composition of the object, whereas standard CT can only visualize its morphology. More information on the physics and practical applications of ME-CT can be found, for example in \cite{Heismann2012,McCollough,So2021}.

Mathematically, ME-CT measurements are modeled by a nonlinear integral transform that maps x-ray attenuation function (or coefficient) of the object to weighted integrals of its x-ray transform over photon energy. Let $\Omega \in \R^d$, $d=2,3$, be the spatial domain of the imaged object. For $y \in \Omega$ and photon energy $E$, we denote by $\mu(y,E)$ the x-ray attenuation coefficient of the object. For $1 \leq i \leq n$, let $w_i(E)$ be the spectral weight function of the $i$-th energy measurement. For example, in the case of energy integrating detectors, $w_i(E)$ is the product of the $i-$th x-ray source energy spectrum $S_i(E)$ and the detector response function $D(E)$. These weights $w_i(E)$ are known and assumed to be compactly supported and normalized so that $\textstyle \int_0^\infty w_i(E) dE=1$. Then, the corresponding ME-CT measurements for a line $l$ are given by the integrals
\begin{align}
I_i(l) = \int_0^\infty w_i(E) e^{-\int_{l}\mu(y,E)dy} dE, \quad 1 \leq i \leq n.
\end{align}

The x-ray attenuation coefficient $\mu(E,y)$ is commonly expressed by a superposition of the (known) elemental x-ray attenuation functions $\mu_j(E)$ weighted by the (unknown) partial density of each respective element $\rho_j(y)$  \cite{AlvarezMacovski,Heismann2012}:
\begin{align*}
\mu(E,y) = \sum_{j=1}^n \mu_j(E)\rho_j(y).
\end{align*}
Here, we assume as many energy measurements as the number of unknown material densities. Then, we can write
\begin{align}\label{ME-CT transform}
  I_i(x) = \int_0^\infty w_i(E) e^{-M(E)\cdot x(l)} dE,
\end{align}
where $M(E) = (\mu_j(E))_{1 \leq j \leq n}$ and $x(l)=(x_j(l))_{1 \leq j \leq n}$ with $x_j(l)= \textstyle \int_l \rho_j dl$ denoting the x-ray transform of $\rho_j$ along a line $l$. We assume that $x=x(l) \in \rR \subset \R^n$ where $\rR \subset \R^n$ is a closed rectangle (a Cartesian product of closed intervals).

%%%%%%
\begin{figure}[htbp]
\begin{center}
   \includegraphics[width=0.8\textwidth]{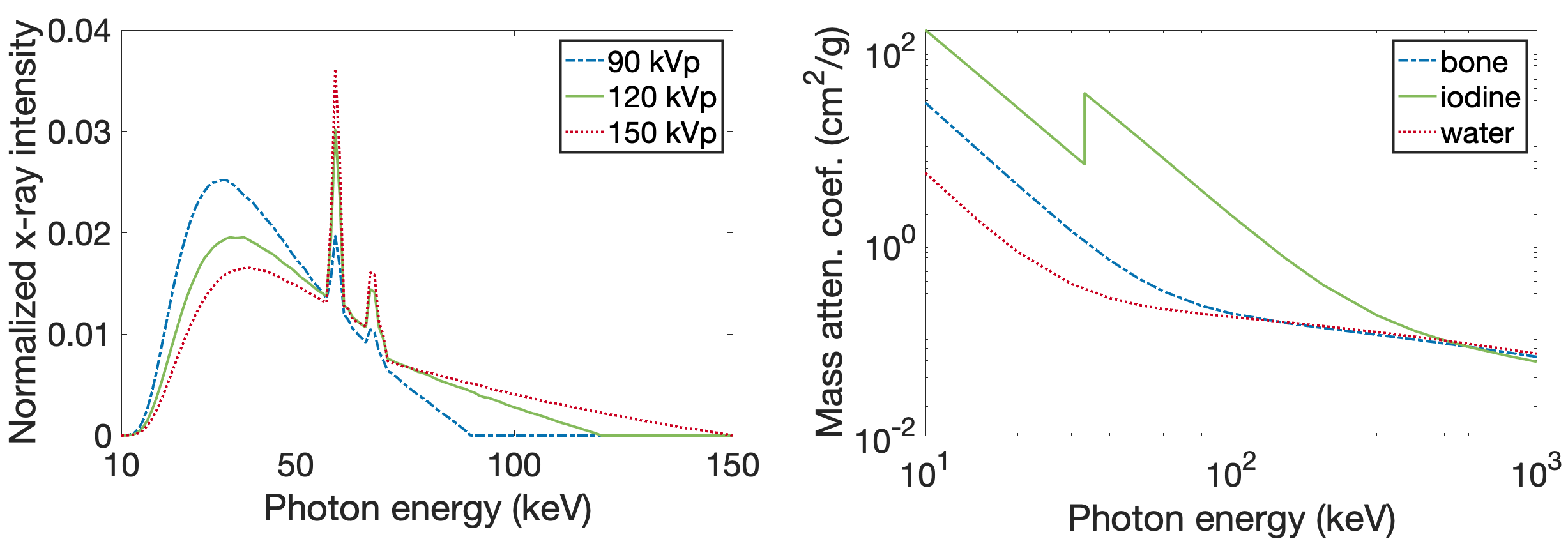}
    \caption{Left: Examples of x-ray source spectrum for varying tube potentials computed using the publicly available code SPEKTR 3.0 \cite{Spektr3}, and then normalized. Right: The x-ray attenuation coefficients of bone, iodine and water as functions of x-ray energy in log-log scale. The raw data was obtained from NIST \cite{NIST}. }
    \label{fig:spectrum_attenuation}
\end{center}
\end{figure}
%%%%%%

Therefore, one way to perform image reconstruction is to perform first a nonlinear inversion reconstructing $x=x(l)$ from $I(x)$ for each line $l$; and then a linear tomographic reconstruction to recover the material density maps $\rho_i$ from their line integrals $x_i(l)$. Since the invertibility of the line integral transform is well-studied, here we focus on the map $g : \rR \subset \R^n \to \R^n$ defined by
\begin{align}\label{logI}
   g(x(l)) = (g_i(x(l)))_{1\leq i\leq n}, \qquad g_i(x) = -\ln I_i(x). 
\end{align}
The invertibility of the maps $g_i$ and $I_i$ are equivalent.

The map $g$ is smooth for $w_i$ compactly supported and $\mu_j$ bounded, and its Jacobian at $x \in \rR$ is given by the matrix $J(x)$ with coefficients
\begin{align}
  J_{ij}(x) = \frac{\partial g_i(x)}{\partial x_j} = e^{g_i(x)} \int_0^\infty w_i(E) \mu_j(E) e^{-M(E)\cdot x} dE ,\qquad 1 \leq i,j \leq n.
\end{align}
We note that $\mu_j(E)\geq 0$ and $w_i(E)\geq0$ for all $1 \leq i, j \leq n$, and that all entries of the Jacobian matrix are strictly positive.

By the inverse function theorem, if $\det J(x) \neq 0$, then the map $g$ is locally injective. For $n=2$, Alvarez \cite{Alvarez19} studied the invertibility of $g$ by testing for zero values of the Jacobian. In a previous work \cite{BalTer20}, we proved that local injectivity of $g$ guarantees injectivity in the whole domain $\rR$. This may not be true for a general map, but it holds in our case because of the positivity of the entries of the Jacobian matrix $J$. On the other hand, for $n \geq 3$, nonvanishing of the Jacobian determinant is not sufficient, and thus we need to impose further conditions on the Jacobian matrix to ensure global injectivity. This is given in the following theorem.

\begin{theorem}[\cite{GaleNikaido}]\label{GaleNikaido}
Let $F : \rR \subset \R^n \to \R^n$ be differentiable on the closed rectangle $\rR$. If the Jacobian $J(x)$ of $F$ is a $P-$matrix for each $x \in \rR$, then $F$ is injective in $\rR$.
\end{theorem}

A matrix $A$ is called a $P-$matrix if all principal minors of $A$ are positive \cite{FiedlerPtak}. Principal minors of an $n\times n$ matrix $A$ are defined as follows. Let $K$ be a subset of $\langle n \rangle= \{1, \dots, n\}$. We denote by $A_K$ the submatrix of $A$ formed by deleting the rows and columns with indices in $K$. The principal minor of $A$ associated to $K$, denoted by $[A]_K$, is the determinant of $A_K$. We let $[A]_{\langle n \rangle} = 1$.

A map $F:\rR \subset \R^n \to \R^n$ is called a $P-$function if for any $x, y \in \rR, x \neq y$, there exists an index $k=k(x,y)$ such that 
$$(x_k-y_k)(f_k(x)-f_k(y))>0.$$
Here $x_k$ and $f_k(x)$ are the $k$-th components of $x$ and $F(x)$, respectively \cite{MoreRheinboldt}.

It is known that \cite[theorem 5.2]{MoreRheinboldt}, a differentiable map $F$ defined on a rectangle $\rR \subset \R^n$ is a $P-$function if its Jacobian $J(x)$ is a $P-$matrix for all $x \in \rR$. Thus, a $P-$function $F$ is injective and hence invertible on its range $F(\rR)$. Moreover, the inverse is a $P-$function as well (\cite[theorem 3.1]{MoreRheinboldt}). 

The class of $P-$matrices includes positive quasi-definite matrices as well as strictly diagonally dominant matrices with positive diagonal entries. According to our numerical experiments in \cite{BalTer20}, the Jacobians in ME-CT are often $P-$matrices for varying spectral weights, but they are neither quasi-definite nor diagonally dominant. If the Jacobian matrix is quasi-definite or strictly diagonally dominant everywhere, then iterative algorithms such as Gauss-Seidel are guaranteed to converge to the global inverse \cite{More1972}. However, in the case of $P-$matrix Jacobians, we were not able to find an algorithm in the literature that is guaranteed to converge. In this paper, we propose such an algorithm and prove its convergence for which we need an estimate for the Lipschitz constant of the inverse map. Such a result was given in our previous paper \cite[Theorem 9]{BalTer20} under conditions on the Lipschitz constant of the forward map that were not explicitly formulated. In this paper, we make these conditions explicit and present new estimates.

\medskip

The paper is organized as follows. Section 2 presents quantitative estimates for injectivity of $P-$functions. In section 3, we propose a damped newton type algorithm for the inversion of $P-$functions. Section 4 contains the application of our results to ME-CT in the two and three materials settings with an equal number of measurements. 

We use the following notation:  $\langle n \rangle = \{1, \dots, n\}$, $\I$: the identity matrix, $\|x\|$: the Euclidean norm, $\|x\|_{\infty} = \max_{i \in \langle n \rangle} |x_i|$, $\norm{A} = \max_{\norm{x} = 1}\norm{Ax}$, and $ \vvvert A \vvvert = \max_{i,j \in \langle n \rangle} |a_{ij}|$.

%%-----------------------------------------------------------------------
\section{Injectivity and its quantitative estimation}
%%-----------------------------------------------------------------------
Let $\rR \subset \R^n$ be a closed rectangle. Suppose that $F: \rR \to \R^n$ is a continuously differentiable map with Jacobian matrix $J(z)$, $z \in \rR$. The Lipschitz constant of $F$ in $\rR$ is given by $L = \textstyle \max_{z \in \rR} \vvvert J(z)\vvvert.$ In this section, we prove that the inverse map $F^{-1}$ is also Lipschitz continuous on $F(\rR)$, and provide a bound for its Lipschitz constant. To this end, we extend the map $F$ to $\R^n$ using $\hF: \R^n \to \R^n$ defined by
\begin{align}\label{Fext}
\hF(x) = F(P(x))+L(x-P(x)),
\end{align}
where $P: \R^n \to \rR$ is the orthogonal projection map given by $\textstyle P(x) = \argmin_{z \in \rR} \|x-z\|$. Note that $P(x)=x$ when $x\in\rR$ while $P(x)\in\partial\rR$ when $x\in\R^n\backslash \rR$. The above extension with $L=1$ was used by Mas-Colell \cite{MasColell} in proving theorem \ref{GaleNikaido} of Gale and Nikaido for polyhedral domains.

Since $F$ and $P$ are Lipschitz continuous, $\hF$ is also Lipschitz continuous, and hence almost everywhere differentiable (by Rademacher's theorem). Moreover, for all $x \in \R^n$ where $\hF$ is differentiable, the Jacobian of $\hF$ is given by
\begin{align} \label{Jhat}
\hJ(x) = J(P(x))DP(x) + L(\I-DP(x)),
\end{align}
with $DP(x)$ being an $n \times n$ diagonal matrix with diagonal entries either 0 or 1. Observe that the Lipschitz constant of $\hF$ is also equal to $L$. We now state our quantitative injectivity estimate.

\begin{theorem}\label{q-injectivity}
 Let $\rR \subset \R^n$ be a closed rectangle. Suppose that $F: \rR \to \R^n$ is a continuously differentiable map with a $P-$matrix Jacobian $J(z)$ at every $z \in \rR$. Then, for all $x, y \in \R^n$, for the extended map $\hF$, we have 
\begin{align}\label{invFh is Lipschitz}
 \|\hF(x)-\hF(y)\|_\infty \geq \tau \|x-y\|_\infty,
\end{align}
where
\begin{align}\label{tau}
\tau = n^{-\frac{1}{2}}(n-1)^{\frac{1-n}{2}} \min_{z \in \rR} \min_{K \subset \langle n\rangle}  L^{|K|+1-n}[J(z)]_K.
\end{align}
In particular, for all $x, y \in \rR$, 
\begin{align}\label{invF is Lipschitz}
 \|F(x)-F(y)\|_\infty \geq \tau \|x-y\|_\infty.
\end{align}
\end{theorem}

The constant $1/\tau$ provides an upper bound for $\hat{M} = \max_{x \in \R^n} \vvvert \hJ(x)^{-1} \vvvert$.

The proof of theorem \ref{q-injectivity} will be given later in this section. Two alternative derivations of \eqref{invF is Lipschitz} with different constants $\tau$ are provided in an Appendix.

\begin{remark}\normalfont
Different extensions than \eqref{Fext} may also be considered.
For instance, when $n=2$ and $J>0$, the constant $L$ may be replaced by any matrix 
$\textstyle \begin{bmatrix}
L & -m \\
-m & L 
\end{bmatrix},$
with 
 $0 \leq m \leq \min_{z \in \rR} \min_{i,j=\{1,2\}} J_{ij}(z)$. Then, $\hF$ is also a $P-$function and $\textstyle \vvvert \hJ(x)^{-1}\vvvert\leq \frac{L}{\det J(P(x))}$. The extension in \eqref{Fext} was chosen to optimize the bound on the norm of the inverse Jacobian.
 \end{remark}

%%%%%%%%%%%%%%%%%%%%%%%%%%%%%%%%%%%%
\subsection{Regularized extension and proof of theorem \ref{q-injectivity}}
Our proof of theorem \ref{q-injectivity} requires a regularized version of the extension \eqref{Fext}. Without loss of generality, we consider $\rR = [0,1]^n$. We first assume that the Jacobian $J(x)$ of $F(x)$ is a $P-$matrix at every $x \in \rR$, so $F$ is globally injective in $\rR$. By continuous differentiability, there exists $\e >0$ such that $J(x)$ is a $P-$matrix at every $x \in \rR_\e = [-\e/2,1+\e/2]^n$.

Let $p_\e: \R \to [-\e/2,1+\e/2]$ be the function defined by
\begin{align}\label{regularized_projection}
 p_\e(x) =
\begin{cases}
-\e/2, & x < -\e,\\
\frac{1}{2} \left( x  + \frac{\e}{\pi}\sin\left(\frac{\pi x}{\e} \right) \right),  & -\e < x <0,\\
x, & 0 \leq x \leq 1,\\
\frac{1}{2} \left( x+1 + \frac{\e}{\pi}\sin\left(\frac{\pi (1-x)}{\e} \right) \right),  & 1 \leq x \leq 1 + \e,\\
1+\e/2, & 1+\e<x.
\end{cases}
\end{align}
For $\e=0.25$, the plots of $p_\e$, and its first and second order derivatives are given in fig. \ref{fig:regularized_projection}.
%%%%%%
\begin{figure}[htbp]
\begin{center}
   \includegraphics[width=0.7\textwidth]{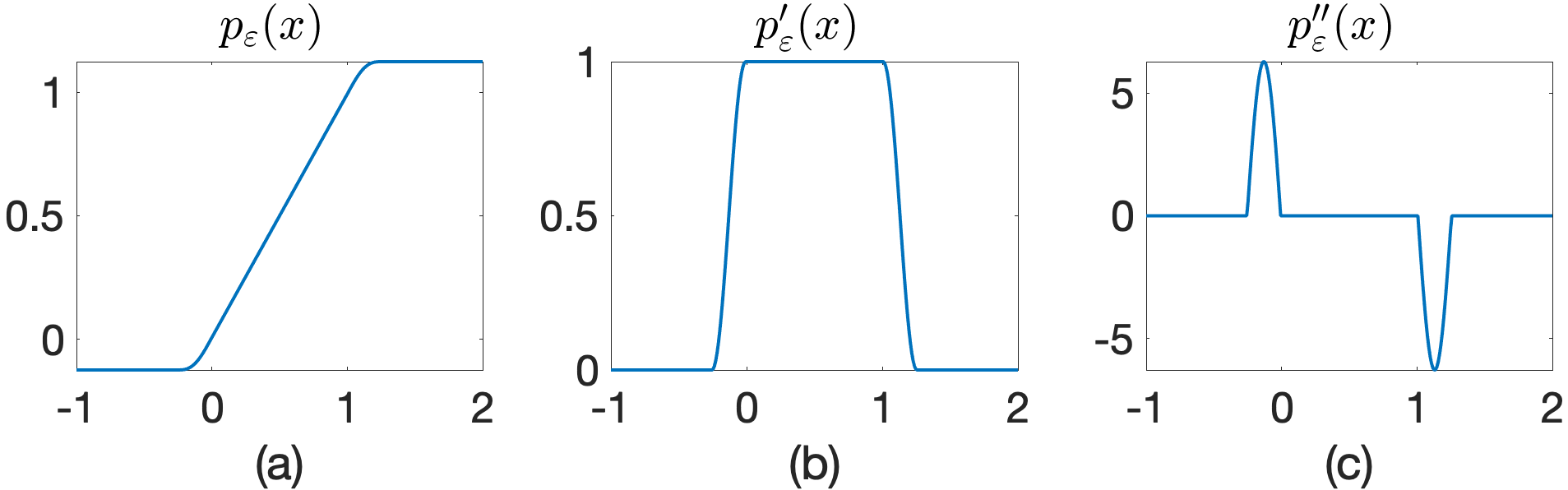}
    \caption{Plots of (a) $p_\e(x)$ given in \eqref{regularized_projection}, and its (b) first and (c) second derivatives when $\e = 0.25$.}
    \label{fig:regularized_projection}
\end{center}
\end{figure}
%%%%%%

Let $L_\e$ be the Lipschitz constant of $F$ in $\rR_\e$, i.e., $L_\e = \max_{z \in \rR_\e} \vvvert J(z)\vvvert$. We define $\hF_\e: \R^n \to \R^n$ by 
\begin{align}\label{Fext_e}
\hF_\e(x) = F(P_\e(x))+L_\e(x-P_\e(x)),
\end{align}
where 
$P_\e(x) = (p_\e(x_i))_{i \in \langle n \rangle}.$ We note that  as $\e \to 0$, we have $P_\e \to P$ and thus  $\hF_\e \to \hF$ pointwise. We then have the following result.

\begin{prop}
$\hF_\e$ extends $F$ from $\rR$ to $\R^n$, and is as smooth as $F$. Moreover, it is a $P-$function, and hence a diffeomorphism.
\end{prop}
\begin{proof}
Since $P_\e(x) = x$ for all $x \in \rR$ and $\e \geq 0$, we have $\hF_\e \big|_{\rR} = F$. By \cite[theorem 5.2]{MoreRheinboldt}, we know that $\hF_\e$ is a $P-$function if its Jacobian $\hJ_\e(x)$ is a $P-$matrix for all $x \in \R^n$. Observe that
\begin{align} \label{Jhat_e}
\hJ_\e(x) = J(P_\e(x))DP_\e(x) + L_\e(\I-DP_\e(x)),
\end{align}
with $DP_\e(x)$ being an $n \times n$ diagonal matrix with diagonal entries $p'_\e(x_i)$, $i \in \langle n\rangle$.

Therefore,
\begin{align}\label{ detJhat_e}
\det(\hJ_\e(x)) = \sum_{K \subset \langle n\rangle} L_\e^{|K|}\Big(\prod_{k \in K} (1-p'_\e(x_k))\Big) \Big(\prod_{k \in \langle n\rangle \setminus K}p'_\e(x_k) \Big) [J(P_\e(x))]_K.
\end{align}

Now since $P_\e(\R^n) = \rR_{\e}$ and $J(x)$ is a $P-$matrix for all $x \in \rR_{\e}$, we have $[J(P_\e(x))]_K >0$ for all $K \subset \langle n\rangle$. Moreover, $0 \leq p'_\e \leq 1$. Thus, all terms in the sum \eqref{ detJhat_e} are nonnegative and do not vanish at the same time, which implies that $\det(\hJ_\e(x)) >0$ for all $x \in \R^n$. Applying the same argument to principal submatrices, we obtain that all principal minors of $\hJ_\e(x)$ are positive for all $x \in \R^n$. Hence, $\hJ_\e(x)$ is a $P-$matrix at every $x \in \R^n$.
\end{proof}

We now present a quantitative injectivity estimate for $\hF_\e$.
\begin{theorem}\label{t:invF is Lipschitz}
For all $x, y \in \R^n$,   
\begin{align}\label{invFhat is Lipschitz}
 \|\hF_\e(x)-\hF_\e(y)\|_\infty \geq \tau_\e \|x-y\|_\infty,
\end{align}
where
\begin{align}\label{tau_e}
   \tau_\e := n^{-\frac{1}{2}}(n-1)^{\frac{1-n}{2}} \min_{z \in \rR_{\e}} \min_{K \subset \langle n\rangle}  L_\e^{|K|+1-n}[J(z)]_K.
\end{align}
\end{theorem}
\begin{proof}
Let $x, y \in \R^n$. Then $\hF_\e(x)=u$ and $\hF_\e(y)=v$ for some $ u, v \in \R^n$. Applying the Mean Value theorem to $\hF_\e^{-1}$, then using Cauchy-Schwarz inequality and the equivalence of norms, we obtain
\begin{align*}
 \|x-y\|_\infty = \|\hF_\e^{-1}(u)-\hF_\e^{-1}(v)\|_\infty  \leq \sqrt{n} \hat{M}_\e \|u-v\|_\infty,
\end{align*}
where $\hat{M}_\e := \max_{x \in \R^n} \vvvert \hJ_\e(x)^{-1}\vvvert$. 

We now need to find an upper bound for $$ (\hJ_\e(x)^{-1})_{ij} = \frac{(-1)^{i+j}[\hJ_\e(x)]_{\{j\},\{i\}}}{\det \hJ_\e(x)},$$ 
for all $x \in \R^n$ and $i, j \in \langle n \rangle$. By Hadamard inequality (see e.g. \cite[p. 1077]{Gradshteyn}), we have 
$$[\hJ_\e(x)]_{\{i\},\{j\}}\leq (n-1)^{\frac{n-1}{2}}\vvvert \hJ_\e(x)\vvvert^{n-1}.$$
Moreover,
\begin{align*}
\vvvert \hJ_\e(x)\vvvert= \max_{i,j \in \langle n \rangle}  \left\{ J_{ij}(P_\e(x))p'_\e(x_j), J_{jj}(P_\e(x))p'_\e(x_j) + L_\e(1-p'_\e(x_j))\right\} \leq L_\e.
\end{align*}

We also have
\begin{align}\label{lb4detJ}
\det(\hJ_\e(x)) 
&= \sum_{K \subset \langle n\rangle} L_\e^{|K|} \Big(\prod_{k \in K} (1-p'_\e(x_k))\Big) \Big(\prod_{k \in \langle n\rangle \setminus K}p'_\e(x_k) \Big) [J(P_\e(x))]_K\nonumber\\
&\geq \Big[\sum_{K \subset \langle n\rangle} \Big(\prod_{k \in K} (1-p'_\e(x_k))\Big) \Big(\prod_{k \in \langle n\rangle \setminus K}p'_\e(x_k) \Big) \Big] \min_{K \subset \langle n\rangle}  L_\e^{|K|}[J(P_\e(x))]_K\\
&= \min_{K \subset \langle n\rangle} L_\e^{|K|}[J(P_\e(x))]_K\nonumber,
\end{align}
as the sum in square brackets is equal to 1. Hence,
\begin{align}\label{invJnorm_e}
\vvvert \hJ_\e(x)^{-1}\vvvert \leq \frac{ (n-1)^{\frac{n-1}{2}} L_\e^{n-1}}{\min_{K \subset \langle n\rangle} L_\e^{|K|}[J(P_\e(x))]_K}.
\end{align}
Since $P_\e(\R^n) = \rR_{\e}$, by maximizing over all $x \in \R^n$, one obtains that $\hat{M}_\e \leq \frac{1}{\sqrt{n} \tau_\e}$, which yields \eqref{invFhat is Lipschitz}.
\end{proof}

\textit{Proof of Theorem \ref{q-injectivity}.} For $\e \to 0$, we have $\hF_\e \to \hF$ and $\tau_\e \to \tau$. Thus, the result follows from theorem \ref{t:invF is Lipschitz}. \hfill $\square$\\

The following estimates will be used in section 3.
\begin{prop}\label{p:estimate_DDF}
Let $DJ$, $D\hJ_\e$, and $D\hJ_\e^{-1}$ denote the Hessian of $F$, $\hF_\e$ and $\hF_\e^{-1}$, respectively. Then,
\begin{align}\label{Bound4HessianFinv}
\vvvert D\hat{J}_\e^{-1}\vvvert\leq 2n\tau_\e^{-2}\vvvert D\hat{J}_\e \vvvert,
\end{align}
and
\begin{align}\label{Bound4HessianF}
\vvvert D\hat{J}_\e \vvvert \leq \vvvert DJ\vvvert + \frac{\pi}{\e}L_\e.
\end{align}
\end{prop}
\begin{proof}
We continue using the notation $\vvvert A\vvvert = \max_{i,j,k\in \langle n \rangle} |a_{ijk}|$.
For $i,j,k \in \langle n \rangle$, and $x \in \R^n$,
 \begin{align*}
(D\hJ_\e(x)_k)_{ij} &= \frac{\partial^2 (\hF_\e)_k (x)}{\partial x_j\partial x_i} 
 = \frac{\partial^2 F_k (P_\e(x))}{\partial x_j\partial x_i} p_\e^\prime(x_i)p_\e^\prime(x_j)
+ \left(\frac{\partial F_k(P_\e(x))}{\partial x_i}-L_\e\delta_{ki}\right)\delta_{ij}p_\e^{\prime\prime}(x_i)\\
&= p_\e^\prime(x_i)p_\e^\prime(x_j) (DJ(P_\e(x))_k)_{ij} 
+ \delta_{ij}p_\e^{\prime\prime}(x_i)\left((J(P_\e(x))_{ki}-L_\e\delta_{ki}\right).
\end{align*}
Since $|p_\e^{\prime}| \leq 1$ and $|p_\e^{\prime\prime}| \leq \frac{\pi}{2\e}$, the estimate \eqref{Bound4HessianF} holds. We note that since $P_\e(\R^n) = \rR_{\e}$ and $F$ is smooth in $\rR_{\e}$, $DJ$ is bounded, hence so is $D\hJ_\e$.

For the estimate \eqref{Bound4HessianFinv}, we differentiate
\begin{align*}
    (\hJ_\e(x)^{-1})_{ij} = \frac{(-1)^{i+j}[\hJ_\e(x)]_{\{j\},\{i\}}}{\det \hJ_\e(x)},
\end{align*}
to obtain
\begin{align*}
    \frac{\partial}{\partial x_k}(\hJ_\e(x)^{-1})_{ij} = \frac{(-1)^{i+j}\frac{\partial}{\partial x_k}[\hJ_\e(x)]_{\{j\},\{i\}}-(\hJ_\e(x)^{-1})_{ij}\frac{\partial}{\partial x_k}\det \hJ_\e(x)}{\det \hJ_\e(x)}.
\end{align*}
For a matrix $A(x) = (a_{ij}(x))_{i,j \in \langle n \rangle}$, by differentiating $\det A = \sum_{i,j=1}^n (-1)^{i+j}a_{ij}[A]_{\{i\},\{j\}}$, we have
\begin{align*}
    \frac{\partial (\det A(x))}{\partial x_k} 
    = \sum_{i,j=1}^n \frac{\partial (\det A(x))}{\partial a_{ij}} \frac{\partial a_{ij}(x) }{\partial x_k}
    = \sum_{i,j=1}^n (-1)^{i+j}[A(x)]_{\{i\},\{j\}} \frac{\partial a_{ij}(x)}{\partial x_k}.
\end{align*}
Thus,
\begin{align*}
    \frac{\partial}{\partial x_k} (\det \hJ_\e(x))
    = \sum_{i,j=1}^n (-1)^{i+j}[\hJ_\e(x)]_{\{i\},\{j\}} \frac{\partial^2 (\hF_\e)_{i}(x)}{\partial x_k \partial x_j},
\end{align*}
and
\begin{align*}
    \frac{\partial}{\partial x_k}  [\hJ_\e(x)]_{\{j\},\{i\}}
    = \sum_{l,m=1}^{n-1} (-1)^{l+m}[\hJ_\e(x)]_{\{j,l'\},\{i,m'\}} \frac{\partial^2 (\hF_\e)_{l'}(x)}{\partial x_k \partial x_{m'}},
\end{align*}
where $l' = l + \delta_{j\leq l}$ and $m' = m + \delta_{i\leq m}$ with $\delta_{j\leq l} = 1$ if $j\leq l$; and 0, otherwise.

By Hadamard's inequality, we have 
$$[\hJ_\e(x)]_{\{i\},\{j\}}\leq (n-1)^{\frac{n-1}{2}}\vvvert \hJ_\e(x)\vvvert^{n-1},$$
and
$$[\hJ_\e(x)]_{\{i,l'\},\{j,m'\}}\leq (n-2)^{\frac{n-2}{2}}\vvvert \hJ_\e(x)\vvvert^{n-2}.$$
Now using the above estimates in
\begin{align}
\left| \frac{\partial}{\partial x_k}(\hJ_\e(x)^{-1})_{ij} \right|
    \leq \frac{\left|\frac{\partial}{\partial x_k}[\hJ_\e(x)]_{\{j\},\{i\}} \right|+ \left|\hJ_\e(x)^{-1})_{ij}\right| \left| \frac{\partial}{\partial x_k}\det \hJ_\e(x) \right|}{\det \hJ_\e(x)},
    \end{align}
and then maximizing over all $x \in \R^n$, in view of \eqref{tau_e} and \eqref{invJnorm_e}, we obtain \eqref{Bound4HessianFinv}.
\end{proof} 

%%----------------------------------------------
\section{An inversion algorithm for $P-$functions}
%%----------------------------------------------
In section 2, we presented stability estimates for a $P-$function $F$ and its extension $\hat{F}_{\e}$. We are now interested in finding an algorithm to solve the inverse problem $F(x) = y$ where $F$ is a $P-$function defined on a rectangular region. We did not find any standard iterative algorithm that is guaranteed to converge to the unique attractor in the $P-$function setting. In this section, we will propose an algorithm taking the form of damped Newton's method which is guaranteed to converge for a smooth extension $\hat F:=\hat{F}_{\e}$ at a fixed value of $\e$. 

\subsection{Existence of periodic orbit for Newton's method}

Consider the inverse problem $F(x)=y^*$ and the standard iterative Newton's method $F(x_n)+DF(x_n)(x_{n+1}-x_n)=y^*$, or equivalently, $x_{n+1} = x_n+ DF^{-1}(x_n)(y^*-F(x_n))$. We now show that such an algorithm is not guaranteed to converge when $F$ is a $P-$function. Indeed, let $A$ be a $P-$matrix and $F(x)=Ax$ linear in the quadrant $x>0$ (i.e., each coordinate positive). 

Let $\hat F(x)$ be the extension \eqref{Fext} (with $\e=0$) to $\R^n$. The Jacobian of $\hat F$ is then piecewise-constant and equal to $A_q$ for $q$ labeling the $2^n$ quadrants (e.g. equal to identity in the quadrant $x<0$ and to $A$ in the quadrant $x>0$). 

By linearity, using $q=q(x)$ the quadrant label to which $x$ belongs, we observe that 
\[
    x_{n+1} = x_n + A_q^{-1}(x_n)(y^*-\hat F(x_n))= x_n + A_q^{-1}(x_n)(y^*- A_q(x_n) x_n) = A_q^{-1}(x_n)y^*.
\]
For a fixed $y^*$, the above right-hand side thus takes a maximum of $2^n$ values. As soon as $x_{n+1}$ belongs to the same quadrant as $x_n$, then $x_{n+j}=x_{n+1}$ for all $j\geq1$ and the algorithm converges. However, it turns out that cycles in the $x_k$ are quite possible and thus prevent the algorithm from converging to the unique solution $x$ of $F(x)=y^*$ (which does not belong to the cycle). 

As a concrete example in dimension $n=3$, consider
\begin{align}\label{P-matrix:A}
    A = \begin{bmatrix}
    6 & 1 & 7 \\ 7 & 3 & 1 \\ 1 & 8 & 3
    \end{bmatrix}, \quad y^* = \begin{bmatrix}
    -9 \\ -4 \\ -6
    \end{bmatrix}.
\end{align}
Here, $F(x) = Ax$ is a $P-$function for $x > 0$ since $A$ is a $P-$matrix. Incidentally, $A$ has positive entries (as do Jacobian matrix for ME-CT models). We extend $F(x)$ using Mas-Colell's extension \eqref{Fext}. For suitable choices of the initial point $x_0$, we observe a periodic trajectory in Figure \ref{periodic_orbit} (a).

The points involved in the above cyclic trajectories all live away from the hyperplanes separating quadrants. Therefore, for $\e>0$ sufficiently small, $\hat F$ and $\hat F_\e$ coincide in the vicinity of the above numerical trajectories. This shows that the Newton algorithm would also fail to converge for the smooth $P-$function $\hat F_\e$, or as a matter of fact for any possibly $C^\infty$ $P-$function equal to the above $\hat F$ in the vicinity of the trajectories. We have not seen such obstructions to the convergence of the Newton algorithm with smooth functionals in the literature. 

\medskip

To avoid the above periodic orbits in the iterative algorithm, let us consider the following damped Newton's method, 
\begin{align}\label{e:discretization_eps=0}
    x_{n+1}  
    = x_n - D\hat{F}(x_n)^{-1}(\hat{F}(x_n)-y^*)h,
\end{align}
where $h\leq1$ is a constant step size. When $D\hat{F}$ is not defined at $x_n$ belonging to (the closure of) more than one quadrant, we choose $D\hat{F}(x_n) = D\hat{F}_{x_{n-1}}(x_n)$, with 
$\hat{F}_{x_{n-1}}$ the $C^2$ function equal to $\hat{F}$ on the rectangular sector where $x_{n-1}$ lives. 

In the piecewise linear example with $A$ defined in (\ref{P-matrix:A}), we actually observe the persistence of periodic orbits for values of $h<1$, for instance $h=0.8$ in  
 Fig.\ref{periodic_orbit}(b). For $h = 0.7$ however, we observe that the discrete trajectory converges to the target point $\hat{F}^{-1}(y^*)$ as in Fig.\ref{periodic_orbit}(c). In our choice of $F$ and $y^*$, we have $\hat{F}^{-1}(y^*) = y^*$.

\begin{figure}[htbp]
    \centering 
\begin{subfigure}{0.3\textwidth}
  \includegraphics[width=\linewidth]{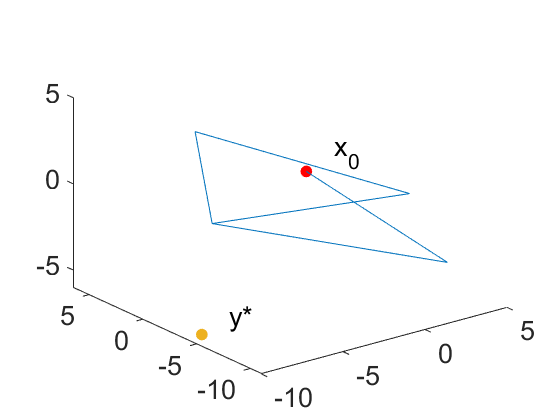}
  \caption{$h=1$}
  \label{fig:1}
\end{subfigure}\hfil 
\begin{subfigure}{0.3\textwidth}
  \includegraphics[width=\linewidth]{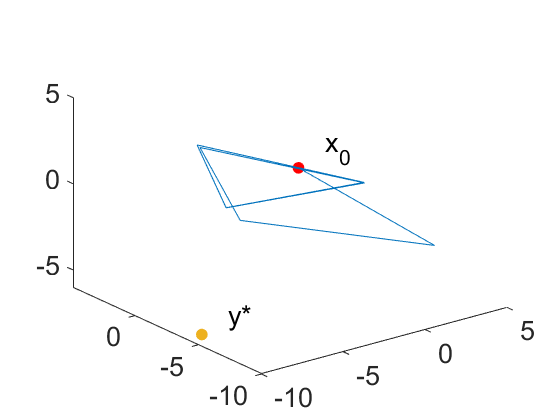}
  \caption{$h=0.8$}
  \label{fig:2}
\end{subfigure}\hfil 
\begin{subfigure}{0.3\textwidth}
  \includegraphics[width=\linewidth]{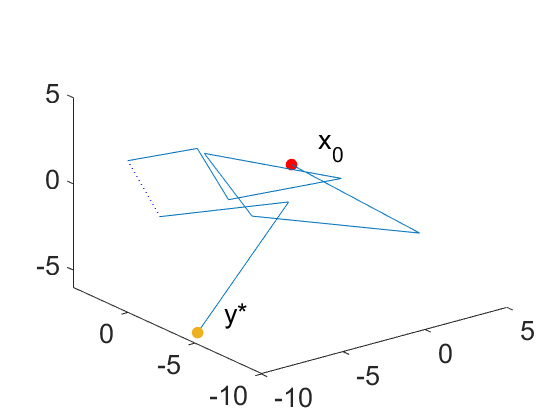}
  \caption{$h=0.7$}
  \label{fig:3}
\end{subfigure}
\caption{Plots of 3D discrete trajectories of the damped Newton algorithm with step sizes (a) h=1; (b) h=0.8; and (c) h=0.7. $x_0$ is the initial point and $y^*=\hat{F}^{-1}(y^*)$ the target point. The dotted line in (c) represents a few steps in the discrete trajectory that are not presented. }
\label{periodic_orbit}
\end{figure}

The above example provide counterexamples to the convergence of the damped algorithm when $h$ is not sufficiently small. We now show that for smooth $P-$functionals such as $\hat F_\eps(x)$, the damped Newton algorithm with $h$ {\em sufficiently small} does indeed converge to the unique solution $\hat F_\eps(x)=y^*$.

\subsection{First-order evolution ODE for injective functions}
%%--------------------------------------------------------------------------
Recall that we are interested in solving the inverse problem $F(x) = y$,  where $F:\ \mathcal{R}\subset \mathbb{R}^n \rightarrow \mathbb{R}^n $ is a continuously differentiable $P-$function defined on rectangle region $\mathcal{R}$.

We showed in section 2 that the $P-$function $F$ defined on a rectangle as well as its extension $\hat{F}_{\e}$ given in \eqref{Fext_e} were injective. We define the associated dynamical system
\begin{align}
    \Psi: \mathbb{R}^n \times \mathbb{R}_+ \rightarrow \mathbb{R}^n, \quad \Psi(y,t) = y^* + e^{-t}(y-y^*), \label{flow in image space}
\end{align}
parametrizing the segment between any $y$ in image space and the target $y^*$.  It is known \cite{MarGorZam94} that for any injective continuous function $F:\Omega \subset \mathbb{R}^n \rightarrow \mathbb{R}^n$ defined on a connected domain $\Omega$, there exists a flow $\Phi:\Omega \subset \mathbb{R}^n \rightarrow \Omega$ satisfying
\begin{align}\label{Phi_def}
    F(\Phi(x,t)) = \Psi(F(x),t),
\end{align}
for all $x \in \Omega, t\geq 0$. When $F$ is differentiable, with fixed initial point $x_0$ (for convenience, we write $\Phi(x_0,t)=x$), we have 
\begin{align*}
    \frac{dF(x)}{dt} &= DF(x) \dot{x}, \qquad 
    \frac{d\Psi(F(x_0),t)}{dt} = \frac{d[F(x^*)+e^{-t}(F(x_0)-F(x^*))]}{dt} = -[F(x) - F(x_*)].
\end{align*}

Since $F(x) = \Psi(F(x_0),t)$, we have $DF(x)\dot{x} = -[F(x)-F(x^*)]$, which implies
\begin{align} \label{origin_ODE}
    \dot{x} = -(DF(x))^{-1}[F(x)-F(x^*)] = -(DF(x))^{-1}[F(x) - y^*].
\end{align}

This yields the following result.
\begin{theorem}[\cite{MarGorZam94}]\label{MG_THM}
Let $F:\Omega \rightarrow \mathbb{R}^n$, where $\Omega \subset \mathbb{R}^n$ is open and connected, be a local diffeomorphism with convex range, and $x^* \in \Omega$, $y^* = F(x^*)$ . The mapping $\Phi$ defined in \eqref{Phi_def} is $C^1$ and it is the flow of the following differential equation
\begin{align}\label{algorithm}
    \dot{x} = -(DF(x))^{-1}[F(x) - y^*] := G(x).
\end{align}
That is, with $x(0) = x_0$,
\begin{align}\label{flow}
    \Phi(x_0,t) = x(t).
\end{align}
\end{theorem}

Equation \eqref{algorithm} provides a continuous version of the damped Newton algorithm in the limit $h\to0$. For any choice of the initial condition $x_0$, the solution $x(t)$ of the dynamical system converges to the desired point $x^*$ (since $F(x(t))$ converges to $y^*$) provided $F$ is an injective differentiable function with convex range. This shows the usefulness of the extension $\hat{F}_{\e}$ since the $P-$function $F$ defined on a rectangle may not have (and indeed does not have in several ME-CT cases) convex range.

\subsection{Convergence of the algorithm for smooth extension}

The damped Newton algorithm, which may be seen as a first-order discretization of \eqref{algorithm}, is given by
\begin{align}\label{discretization}
    x_{n+1} = x_n + G(x_n)h = x_n - D\hat{F}_{\e}(x_n)^{-1}(\hat{F}_{\e}(x_n)-y^*)h, 
\end{align}
with $h>0$ the step size. We now show that for $h$ sufficiently small, the above algorithm converges with linear rate of convergence.  More precisely, we have:
\begin{theorem}\label{euler-cvg}
Let $F: \mathcal{R} \subset \mathbb{R}^n \rightarrow \mathbb{R}^n$ be a $C^2$ $P-$function defined on rectangle region $\mathcal{R}$. Let $\hat{L}_{\eps}$ be the Lipschitz constant of $\hat{F}_{\eps}$ and define $c_1 = \max_{x\in \R^n}\norm{D\hat{F}_{\e}(x)^{-1}}$, $c_2 = \max_{x\in\R^n,\norm{v}=1}\norm{D(D\hat{F}_{\e}(x)^{-1})v}$, and $c_3 = \norm{\hat{F}_{\e}(x_0) - y^*}$. Then there exists $r>0$, such that for all $h<H$ and $i$, 
\begin{align}\label{multi-step-rate}
    \frac{\norm{\hat{F}_{\e}(x_{i+1})-y^*}}{\norm{\hat{F}_{\e}(x_{i})-y^*}} < 1-rh,
\end{align}
where $H = \mathcal{O}(\e)$. In particular, $\hat{F}(x_i)$ converges to $y^*$ with linear rate.
\end{theorem}

The proof follows from the error analysis of the Euler method for the above dynamical system. The constant for the linear convergence rate $r$ may be chosen as, e.g., $r=\frac18$ and the step size $H$ is linear in $\e$; see the proof for a more explicit dependency.

We first prove the following lemma.
\begin{lemma}\label{truncation_error}
Let $x_1$ be the solution to the discretized ODE \eqref{discretization} starting at $x_0$. Let $x(t)$ be the solution to the ODE \eqref{algorithm} with initial condition $x(0) = x_0$. Let $t_1 = h$. Then, we have
\begin{align*}
    \norm{x(t_1)-x_1} \leq \frac{\norm{DG}  \norm{G}}{2}h^2 
\end{align*}
where $\norm{DG} = \max_{t \in [0,t_1],\norm{v}=1} \norm{DG(x(t))v}$ and $\norm{G} = \max_{t \in [0,t_1]} \norm{G(x(t))}$.

\end{lemma}
\begin{proof}
Using Taylor's formula, we expand $x(t)$ at time $t_0$ and evaluate at $t_1=h$ to obtain
\begin{align}\label{x(t_1)}
    x(t_{1}) &= x(t_0) + h\dot{x}(t_0) + \frac{h^2}{2}\ddot{x}(\tau)
    = x(t_0) + hG(x(t_0)) + \frac{h^2}{2}DG(x(\tau))G(x(\tau)). 
\end{align}
By definition of $x_1 = x_0 + hG(x_0)$,
\begin{align*}
    x_{1} - x(t_{1}) &= x_0 + hG(x_0) - x(t_{1})
    = - \frac{h^2}{2}(DG(x(\tau))G(x(\tau))).
\end{align*}
Therefore, $\norm{x(t_{1})-x_{1}} \leq  \frac{h^2}{2}\norm{DG(x(\tau))G(x(\tau))} \leq \frac{\norm{DG}\norm{G}}{2}h^2.$
\end{proof}

\begin{proof} (Theorem \ref{euler-cvg}).
We fix $i$ and define $\Delta_i = \norm{\hat{F}_{\eps}(x_i)-y^*}$. Let $x(t) = \Phi(x_i,t)$ where $\Phi(x,t)$ is defined in \eqref{flow} so that $x(0) = x_i$.  Define $E_i=\norm{x_{i+1}-x(h)}$ for convenience.

Since $\hat{F}_{\eps}(x(h))-y^* = e^{-h}(\hat{F}_{\eps}(x_i)-y^*)$ as we see from \eqref{flow in image space} and \eqref{Phi_def}, and $e^{-h} < 1-\frac{1}{4}h$ for $h$ small (in fact, it is true for 0<h<1), we obtain using the Lipschitz continuity of $\hat{F}_{\eps}$ that
\begin{align}\label{linear eh}
     \frac{\norm{\hat{F}_{\eps}(x_{i+1})-y^*}}{\norm{\hat{F}_{\eps}(x_{i})-y^*}} \leq \frac{\norm{\hat{F}_{\eps}(x_{i+1})-\hat{F}_{\eps}(x(h))} + \norm{\hat{F}_{\eps}(x(h))-y^*}}{\Delta_i} 
    \leq  \frac{\hat{L}_{\eps}E_i + (1-\frac{1}{4}h)\Delta_i}{\Delta_i}. 
\end{align}
Note that $\norm{\hat{F}_{\e}(x(t)) - y^{*}} \leq \Delta_i$ since $\hat{F}(x(t))$ lives on a straight line. We use induction to show $\Delta_k \leq c_3$ for any $k$. For $k=0$, we have $\Delta_0 = c_3$, which provides the base case.  We assume after step $k$ (that is we start with $x_k$) that $\Delta_k \leq c_3$, so that
\begin{align*}
    \norm{G(x(t))} &\leq \norm{D\hat{F}_{\e}^{-1}}\norm{\hat{F}_{\e}(x(t))-y^*} \leq c_1\Delta_k \\
    \norm{DG(x(t))} &= \norm{D(D\hat{F}_{\e}(x(t))^{-1})(\hat{F}_{\e}(x(t))-y^*) + \I} \leq c_2\Delta_k + 1 \leq c_2c_3 + 1.
\end{align*}

By Lemma \ref{truncation_error} we have
\begin{align}\label{case1 Ei}
    E_k
    \leq \frac{\norm{G}\norm{DG}}{2}h^2
    \leq \frac{c_1\Delta_k (c_2c_3+1)}{2}h^2 
\end{align}

By \eqref{linear eh} and \eqref{case1 Ei}
\begin{align}\label{smooth-step-analysis3}
    \frac{\norm{\hat{F}_{\e}(x_{k+1})-y^*}}{\norm{\hat{F}_{\e}(x_{k})-y^*}}
    \leq \frac{\hat{L}_{\eps}E_k + (1-\frac{h}{4})\Delta_k}{\Delta_k} 
    \leq 1-\frac{1}{8}h . 
\end{align}
The inequality \eqref{smooth-step-analysis3} holds with $r=\frac18$ when $h$ sufficiently small. In fact, we can choose $h = \frac{1}{4}\frac{1}{\hat{L}_{\eps}c_1(c_2c_3 + 1)}$. By Proposition \ref{p:estimate_DDF}, $c_2 = \mathcal{O}(\frac{1}{\e})$, so $h = \mathcal{O}(\e)$ necessarily. This shows that $\Delta_{k+1} < \Delta_{k} < c_3$. By induction, we conclude $\Delta_k < c_3$ for any $k$.  

With $\Delta_i < c_3$, the result follows from \eqref{case1 Ei} and \eqref{smooth-step-analysis3} by replacing $k$ with $i$.
\end{proof}

\begin{remark}[Quadratic rate of convergence]\normalfont
Theorem \ref{euler-cvg} shows that the algorithm \eqref{discretization} converges to the attractor $x^* = \hat{F}_{\e}^{-1}(y^*)$ with linear convergence rate. When close enough to $x^*$, we may in fact switch to a standard Newton's method
\begin{align*}
    x_{n+1} = x_n - D\hat{F}_{\eps}(x_n)^{-1} (\hat{F}_{\eps}(x_n)-y^*),
\end{align*}
to achieve quadratic convergence rate. The radius of convergence of Newton's method \cite{Rhe78} is $\rho < \frac{2}{3\beta\gamma}$ where $\beta = \norm{D\hat{F}_{\e}(x^*)^{-1}}$ and $\gamma$ satisfies $\norm{D\hat{F}_{\e}(y)-D\hat{F}_{\e}(x)} \leq \gamma\norm{y-x}, \forall x,y\in \mathcal{R}$. 

We may therefore choose $T$ so that $\frac{1}{\tau_{\e}}\norm{y(T)-y^*} < \rho$, where $\tau_{\e}$ is defined by \eqref{tau_e}. In particular, 
\begin{align}\label{eq:Tcrit}
    T &\geq \log\frac{3\beta\gamma\norm{y_0-y^*}}{2\tau_{\e}}.
\end{align}

The above results therefore provide an algorithm with overall {\em quadratic} rate of convergence. We solve the damped Newton algorithm for a finite time $T$ as given in \eqref{eq:Tcrit} (i.e., for a finite number of steps given by $T/h$) to obtain an approximation $y(T)$ sufficiently close to $y^*$. We then switch to a standard Newton algorithm (with $h=1$) whereby obtaining a quadratic rate of convergence to the fixed point $x^*$.
\end{remark}

\subsection{Remarks on the smoothness of the extension $\hat F$}
%%--------------------------------------------------------------------------
The above results show the unconditional convergence of the damped Newton algorithm when $F$ is extended to $\hat F$ and $\hat F$ is sufficiently smooth. When $\hat F=\hat F_\e$ a smooth extension of $F$, then $h$ needs to be chosen of order $\e$. This is not a major constraint in practice since $\e$ may in fact be chosen reasonably large without significantly modifying the stability of the inversion procedure.

\medskip

Consider the two-dimensional $P-$function
\begin{align*}
    F(x,y) = (x, kx^2+y),
\end{align*}
defined on rectangle region $\mathcal{R}:[-1,1]\times[-1,1]$. The extension $\hat{F}$ on $\R^2$ is
\begin{align*}
    \hat{F}(x,y) = \begin{cases}
    (x,y+k) & x<-1, x>1\\
    (x,kx^2+y) & -1 \leq x \leq 1
    \end{cases}
\end{align*}
In Fig.\ref{p:nonconvex_domain}(a), we display the trajectory associated with the damped Newton algorithm for the extension $\hat F$ and for a very small value of $h$. We observe that this trajectory mostly lives outside of the initial rectangular domain.
\begin{figure*}
        \centering
        \begin{subfigure}[b]{0.4\textwidth}
            \centering
            \includegraphics[width=\textwidth]{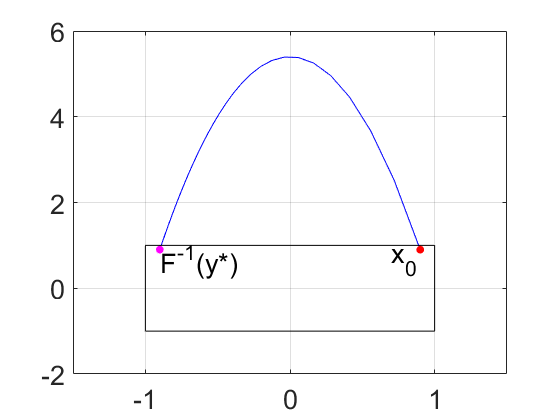}
            \caption{}
            \label{nonconvex-eg1}
        \end{subfigure}
        \begin{subfigure}[b]{0.4\textwidth}  
            \centering 
            \includegraphics[width=\textwidth]{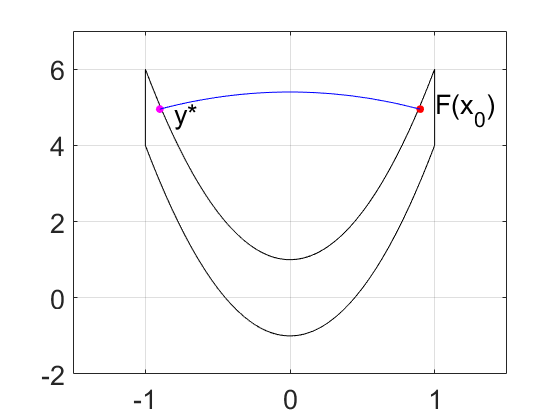}
            \caption{}
            \label{nonconvex-eg2}
        \end{subfigure}
        \caption{The plots of discrete trajectories (blue) of the algorithm in (a) physical space and (b) image space with $k=5, h=0.1$. $x_0$ is the initial point of the algorithm and $y^*$ is the target point in image space.}
        \label{p:nonconvex_domain}
\end{figure*}
This shows that the trajectories of the damped Newton algorithm leave the original domain $\mathcal{R}$ and thus require $F$ to be extended first. There may still be different algorithms allowing us to converge to $F^{-1}(y^*)$ with discrete dynamics staying inside $\mathcal{R}$ but this is not what we considered here.

It would be interesting to see how the algorithm behaves when $\hat F$ is the unregularized extension \eqref{Fext}, or whether $h$ may be chosen independently of $\e$ if the latter is small. We do not have a full theoretical understanding of this case, and in particular do not know if the damped Newton algorithm always converges for the extensions in \eqref{Fext}. 

Here, we provide a family of examples displaying the difficulty to obtain such a convergence result. In particular, we show that the algorithm is not necessarily a contraction at each step in the image space no matter how small $h$ is chosen. This is in sharp contrast with the proofs of convergence of iterative methods for functions with positive definite Jacobians \cite{Rhe78}.

Consider the linear map
\begin{align*}
    F(x,y) &= (x+ky,y), \qquad k > 0,
\end{align*}
defined on rectangular domain $[0,1]\times [0,1]$. The Jacobian matrix of $F$ is a $P-$matrix everywhere inside the domain, so $F$ is a $P-$function. The extension $\hat{F}$ on $\R^2$ is given by
\begin{align*}
    \hat{F}(x,y) = \begin{cases}
    (x,y) & y<0 \\
    (x+ky,y) & 0 \leq y \leq 1 \\
    (x+k,y) & y>1.
    \end{cases}
\end{align*}

For concreteness, choose a initial point $\textstyle x_0 = (\frac{31}{32},\frac{31}{32})$, $ \textstyle y^* = F(x_0) + (\frac{1}{16},\frac{1}{16})$, and $k = 20$. Let $d_i$ be the distance between $\hat{F}(x_i)$ and $y^*$ in the image space. Then, for $h=0.1$, we find $d_6 = 0.0470$, $d_7 = 0.0644$, $d_8 = 0.0579$, $d_{9} = 0.0522$, and $d_{10} = 0.0469$. The increase of the distance at step 7 is due to the crossing of the discontinuity of the Jacobian at the boundary $y=1$. We easily verify that the distance increases when crossing such a singular interface no matter how small $h$ is chosen. The algorithm still converges eventually. We were able to prove (details not shown) for two-dimensional extensions $\hat F$ that the algorithm is always contracting after $m$ steps, i.e., $\|F(x_{mn})-y^*\|$ decreases with $n$ for an appropriately chosen $m$ (equal to $3$ in the above example). The damped Newton algorithm is therefore convergent. By appropriately choosing $k$, we can force $m$ to be as large as we want. 
We were not able to extend the derivation to higher dimensions because of the complexity of the discrete dynamics in the vicinity of the singularities of the Jacobian of the extension $\hat F$. 

%%----------------------------------------------
\section{Application to Multi-energy CT}
%%----------------------------------------------
In this section, we present numerical experiments for multi energy CT transforms with two and three commonly used materials and an equal number of energy measurements. The configuration of parameters was done as follows.
\begin{itemize}
\item The diagnostic energy range $10 - 150$ keV was considered.
\item The energy spectra $S_i, \; i=1,\dots n,$ corresponding to tube potentials $tp_i$ were computed using the publicly available code SPEKTR 3.0 \cite{Spektr3}. For practical purposes, only integer valued tube potentials ranging from 40-150 kVp were considered. We denote $tp = (tp_1, \dots, tp_n)$. We assume that the detectors have linear sensitivity, i.e., $D(E)=E$, which is the case for energy integrating detectors.
\item  The domain of the transform $I$ is chosen as
\begin{align*}
\rR = \Bigg\{ (x_1,\dots,x_n) \in \R_+^n : \; 0 \leq x_j \leq \frac{16}{\displaystyle \max_{10\leq E \leq 150} M_j(E)} \Bigg\},
\end{align*}
where $M_j(E)$ denotes the energy-dependent mass-attenuation of the $j$-th material, and $M(E) = (M_j(E))_{1\leq j\leq n}$.
We note that then $e^{-M(E)\cdot x} \geq e^{-16}$, which is more conservative than necessary in practice.
\end{itemize}

 \subsection{DE-CT: Two materials - two measurements}
 %%--------------------------------------------------------------------
In the following, we consider three different material pairs bone-water, iodine-water, and bone-iodine in the said order, and present the corresponding plots of
\begin{align}\label{hatM2d}
    \hat{M} = \max_{z \in \rR} \vvvert J(z)^{-1}\vvvert= \max_{z \in \rR} \frac{\vvvert J(z)\vvvert}{|\det J(z)|},
\end{align}
and its estimate (see theorem \ref{q-injectivity})
\begin{align}\label{hatM2d_est}
    \hat{M}_{est} = \frac{L}{\min_{z \in \rR} |\det J(z)|}.
\end{align}

%%%Bone-Water
%%%---------
For bone-water material pair, both $\hat{M}$ and $\hat{M}_{est}$ attained their minimum at $(tp_1,tp_2) = (40,150)$, which are 9.03 and 15.95, respectively. A not so good choice for the tube potentials is $(tp_1,tp_2) = (135,150)$ where $\hat{M}$ equals $117.4$ (see fig. \ref{fig:BoneWater}). One would expect the choice $(tp_1,tp_2) = (40,150)$ to lead a better posed problem.
%%%%%%
\begin{figure}[htbp]
   \centering
    \begin{subfigure}{0.21\textwidth}
    \includegraphics[width=\textwidth]{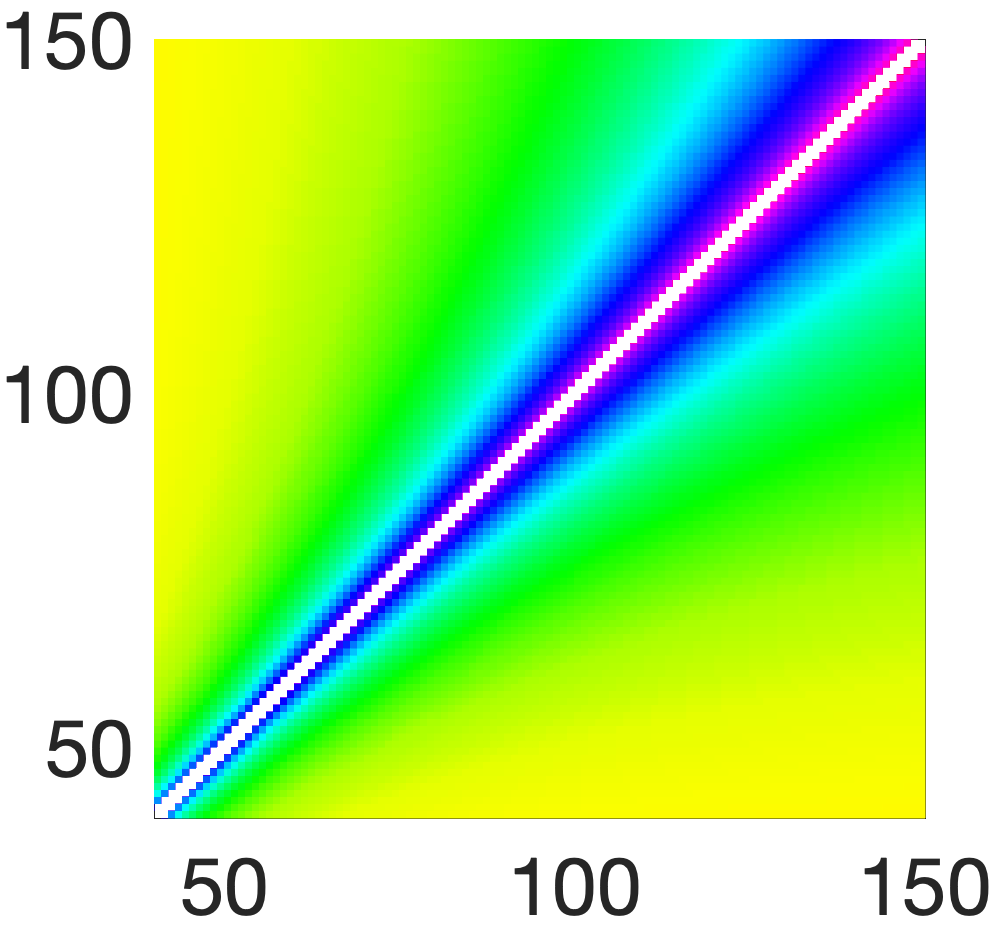}
    \caption{}
    \end{subfigure}%\hspace{-3em}
     \begin{subfigure}{0.25\textwidth}
    \includegraphics[width=\textwidth]{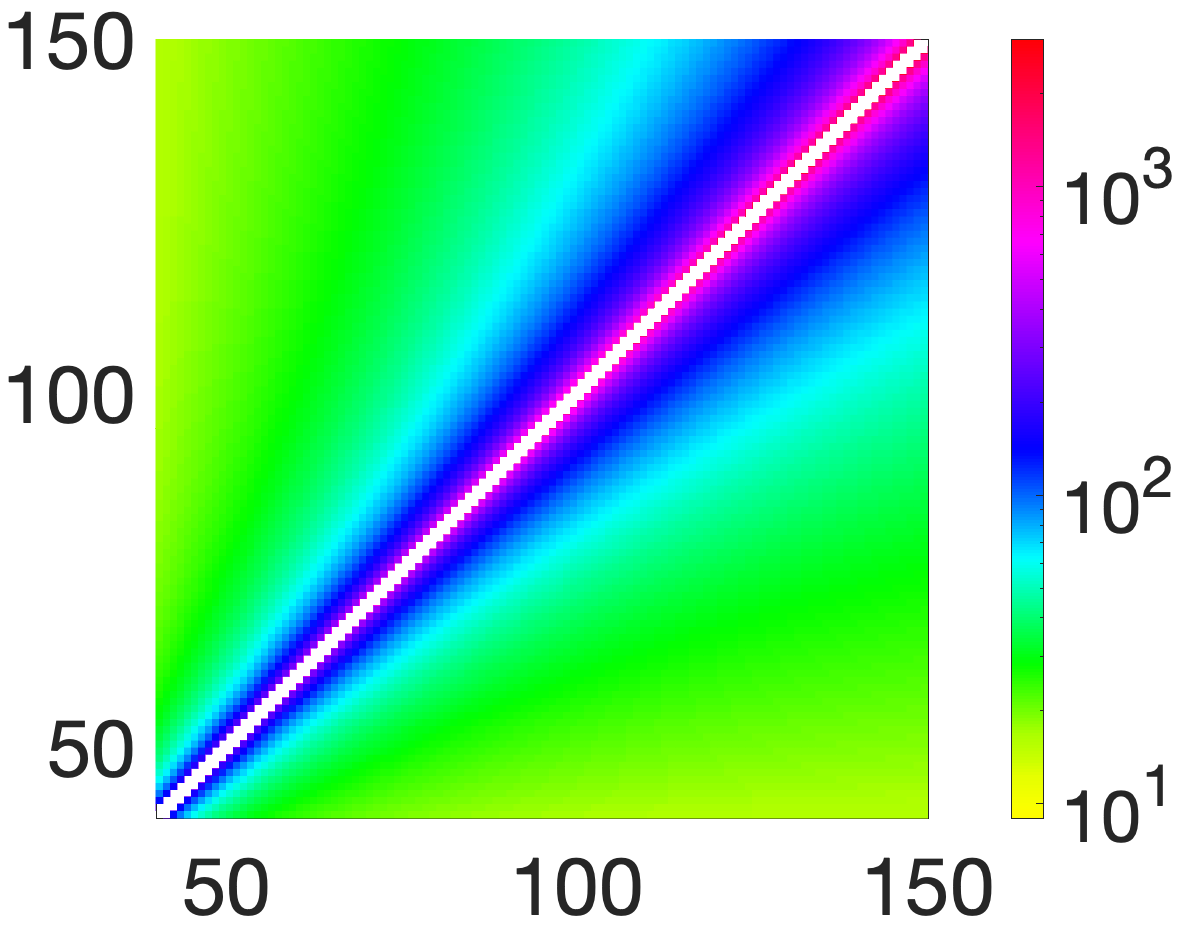}
    \caption{}
    \end{subfigure}%\hspace{-3em}
     \begin{subfigure}{0.22\textwidth}
    	\includegraphics[width=\textwidth]{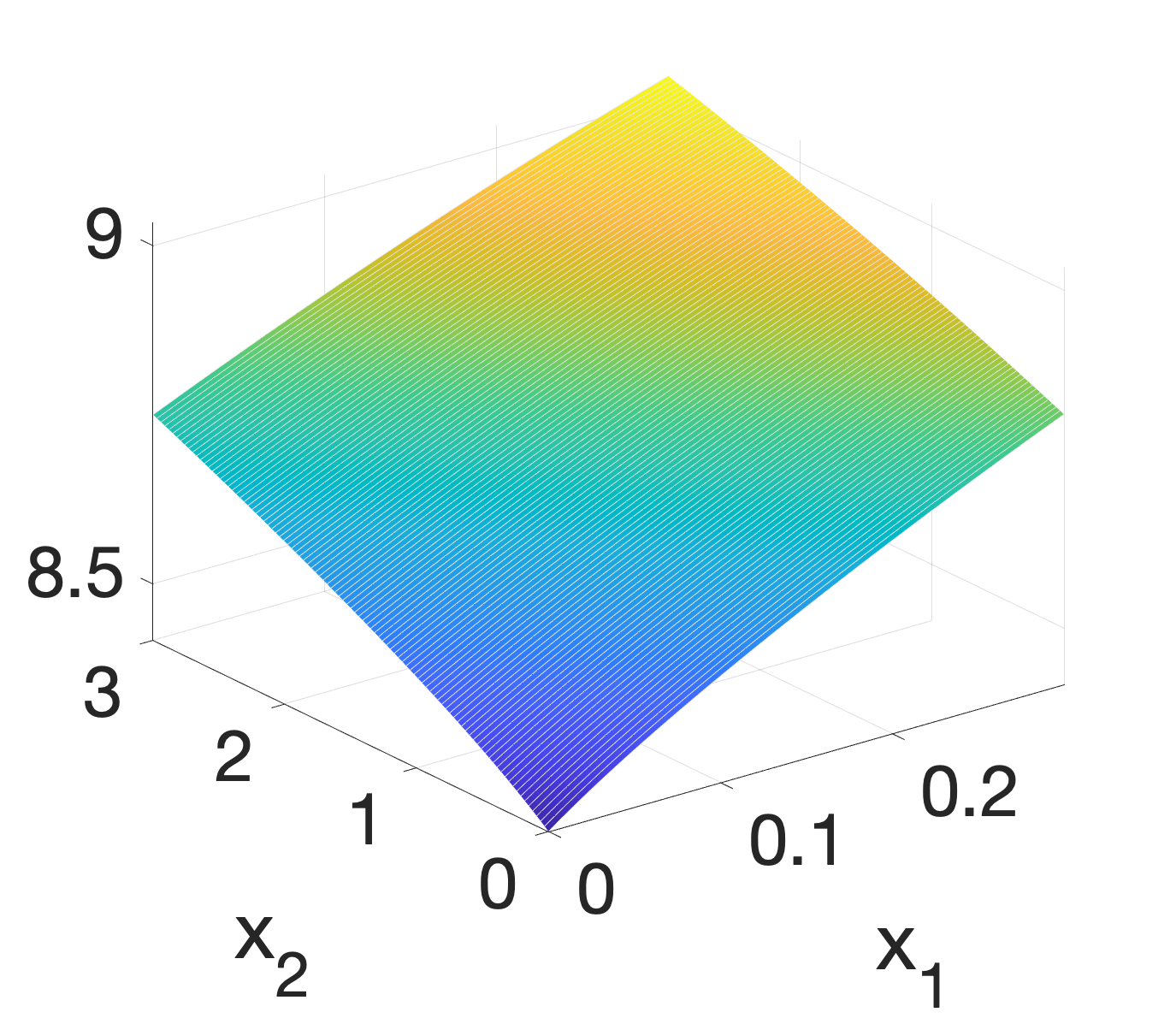}
    	\caption{}
     \end{subfigure}%\hspace{-1em}
          \begin{subfigure}{0.22\textwidth}
    	\includegraphics[width=\textwidth]{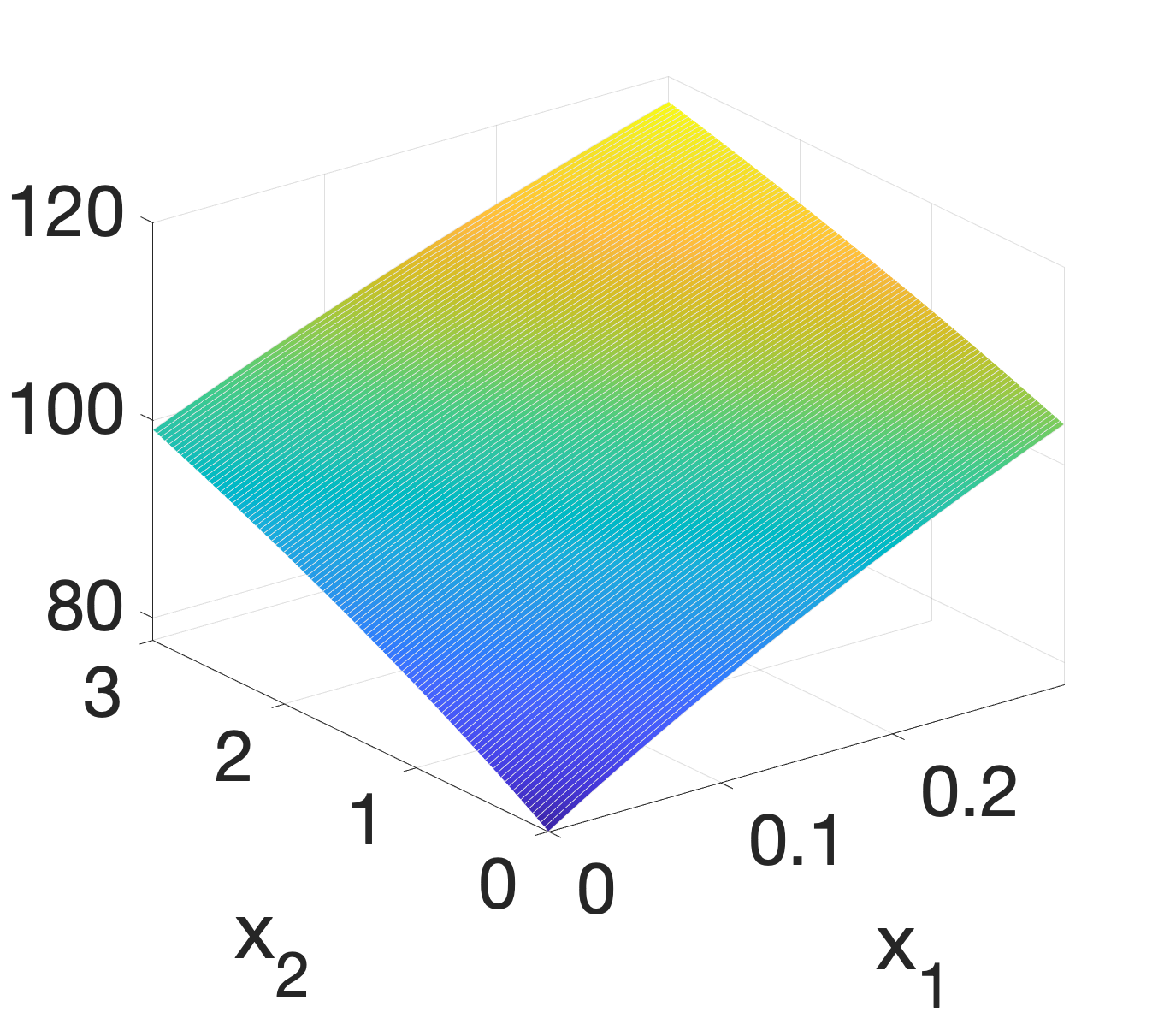}
    	\caption{}
     \end{subfigure}
    \caption{The plots of (a) $\hat{M}$ and (b) $\hat{M}_{est}$ as a function of tube potentials for the material pair (bone,water) in the rectangle $\rR = [0,0.3]\times[0,3]$. The plots of $\vvvert J(x)^{-1}\vvvert$ on $\rR$ when (c) $tp = (40,150)$ and (d) $tp = (135,150)$.}
    \label{fig:BoneWater}
\end{figure}

%%%Iodine-Water
%%%------------
For iodine-water material pair, $\hat{M}$ attained its minimum at $(tp_1,tp_2) = (40,68)$, which is 8.15. $\hat{M}_{est}$ attained its minimum at $(tp_1,tp_2) = (40,66)$, and is 11.3. We note that $\hat{M}(40,66) =8.2$. On the other hand, a poor choice for the tube potentials would be $(tp_1,tp_2) = (55,82)$ where $\hat{M}$ equals $78865$ (see fig. \ref{fig:IodineWater}). We will demonstrate in section \ref{s:reconstruction} that the former indeed leads to considerably better reconstructions.
%%%%%%%%%%%%%%%%%
\begin{figure}[htbp]
   \centering
    \begin{subfigure}{0.21\textwidth}
    \includegraphics[width=\textwidth]{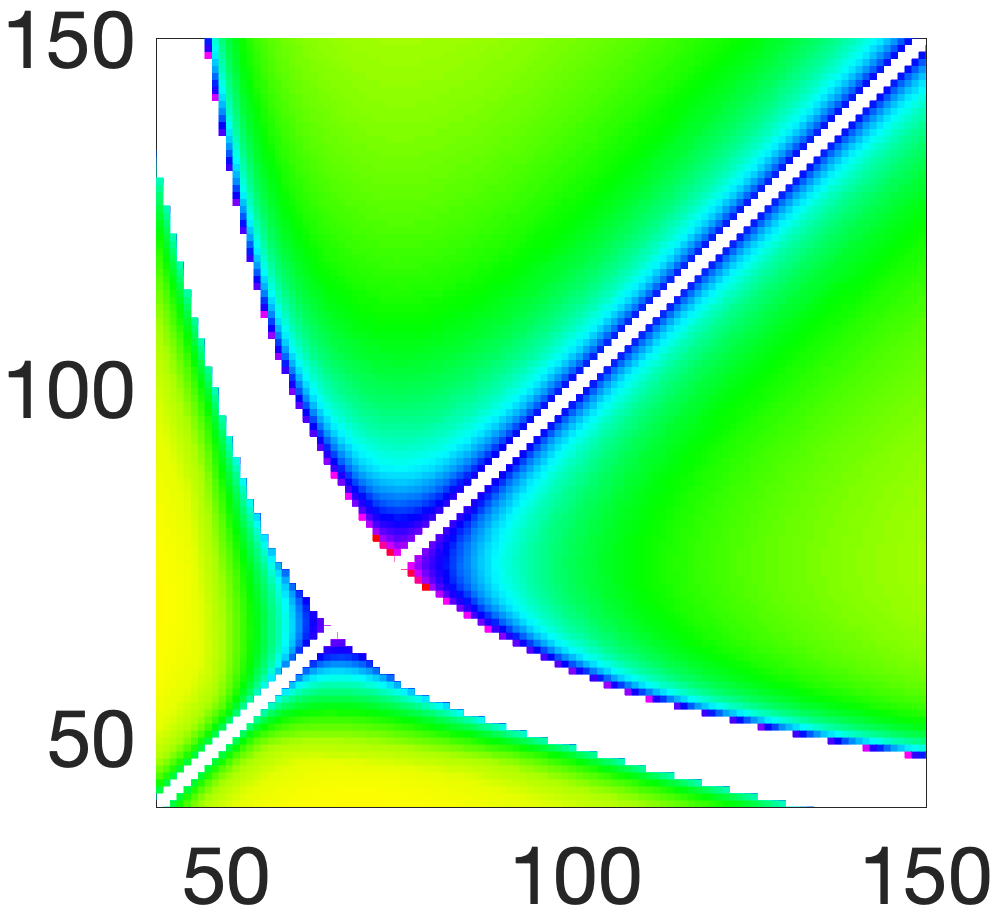}
    \caption{}
    \end{subfigure}%\hspace{-3em}
     \begin{subfigure}{0.25\textwidth}
    \includegraphics[width=\textwidth]{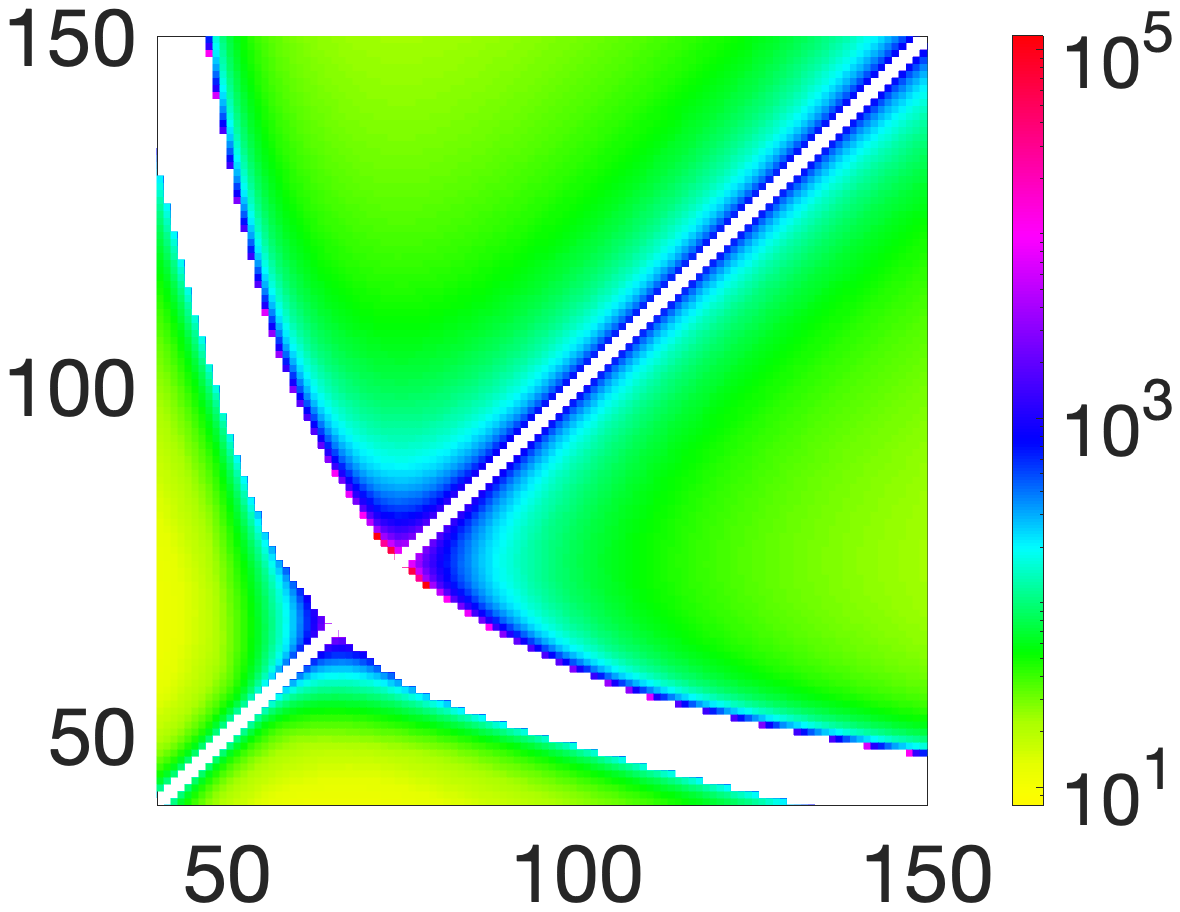}
    \caption{}
    \end{subfigure}%\hspace{-3em}
     \begin{subfigure}{0.22\textwidth}
    	\includegraphics[width=\textwidth]{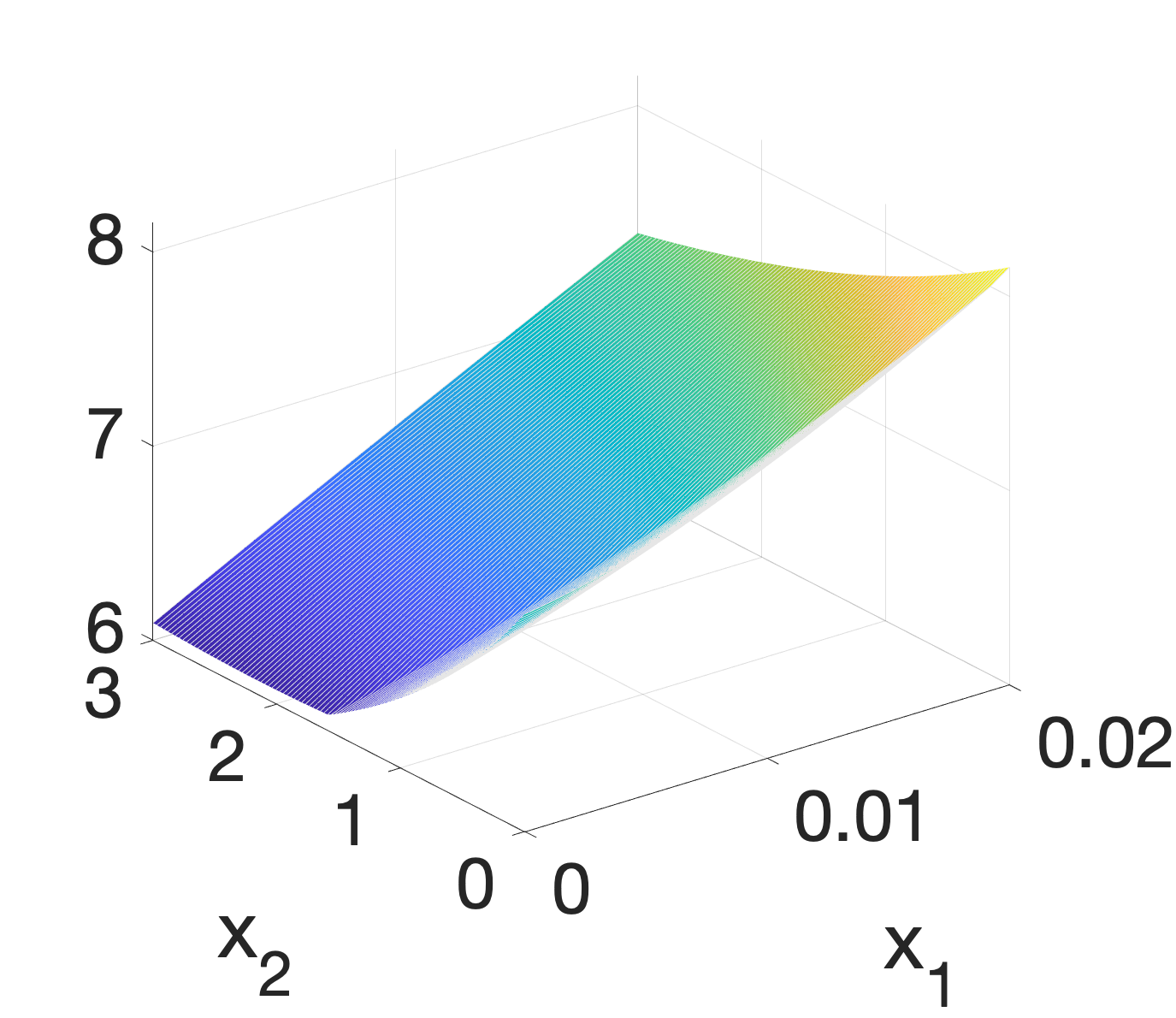}
    	\caption{}
     \end{subfigure}%\hspace{-1em}
          \begin{subfigure}{0.22\textwidth}
    	\includegraphics[width=\textwidth]{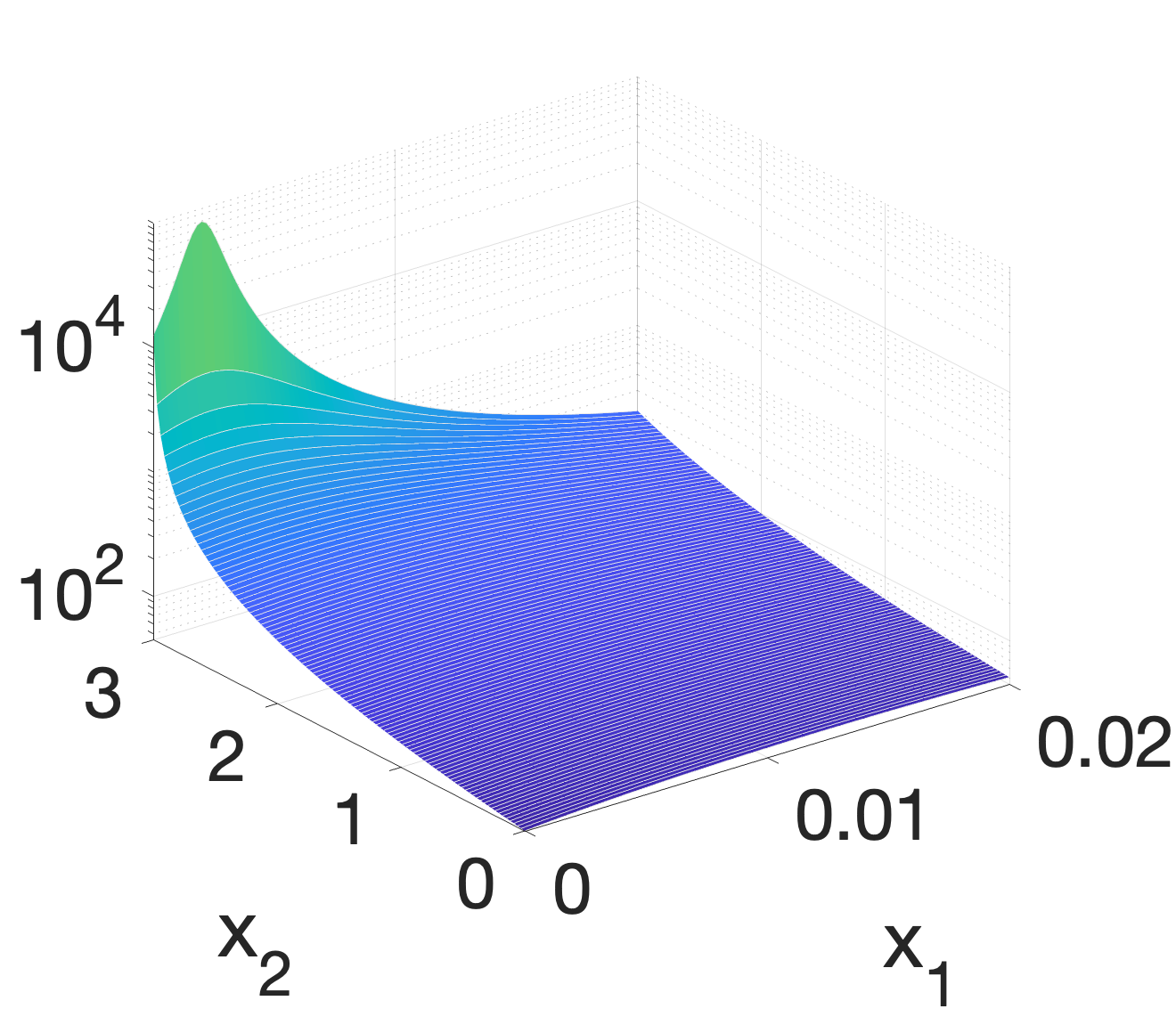}
    	\caption{}
     \end{subfigure}
    \caption{The plots of (a) $\hat{M}$ and (b) $\hat{M}_{est}$ as a function of tube potentials for the material pair (iodine,water) in the rectangle $\rR = [0,0.02]\times[0,3]$. The plots of $\vvvert J(x)^{-1}\vvvert$ on $\rR$ when (c) $tp = (40,68)$ and (d) $tp = (55,82)$.}
    \label{fig:IodineWater}
\end{figure}

%%%Bone-Iodine
%%%-----------
For bone-iodine material pair, $\hat{M}$ attained its minimum at $(tp_1,tp_2) = (40,74)$, which is 0.68. $\hat{M}_{est}$ attained its minimum at $(tp_1,tp_2) = (40,67)$, which is 0.99. We note that $\hat{M}(40,67) =0.71$. A poor choice for the tube potentials would be $(tp_1,tp_2) = (79,132)$ where $\hat{M}$ equals $30536$ (see fig. \ref{fig:BoneIodine}).
%%%%%%%%%%%%%%%%%
\begin{figure}[htbp]
   \centering
    \begin{subfigure}{0.21\textwidth}
    \includegraphics[width=\textwidth]{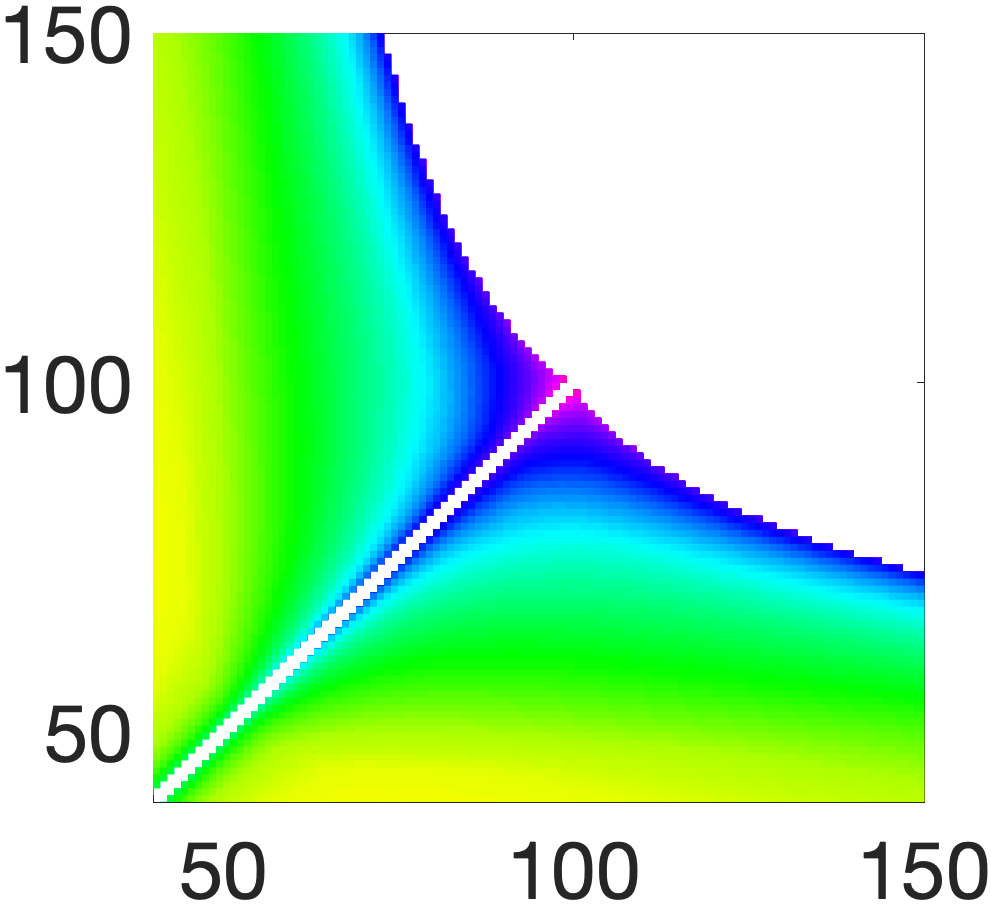}
    \caption{}
    \end{subfigure}%\hspace{-3em}
     \begin{subfigure}{0.25\textwidth}
    \includegraphics[width=\textwidth]{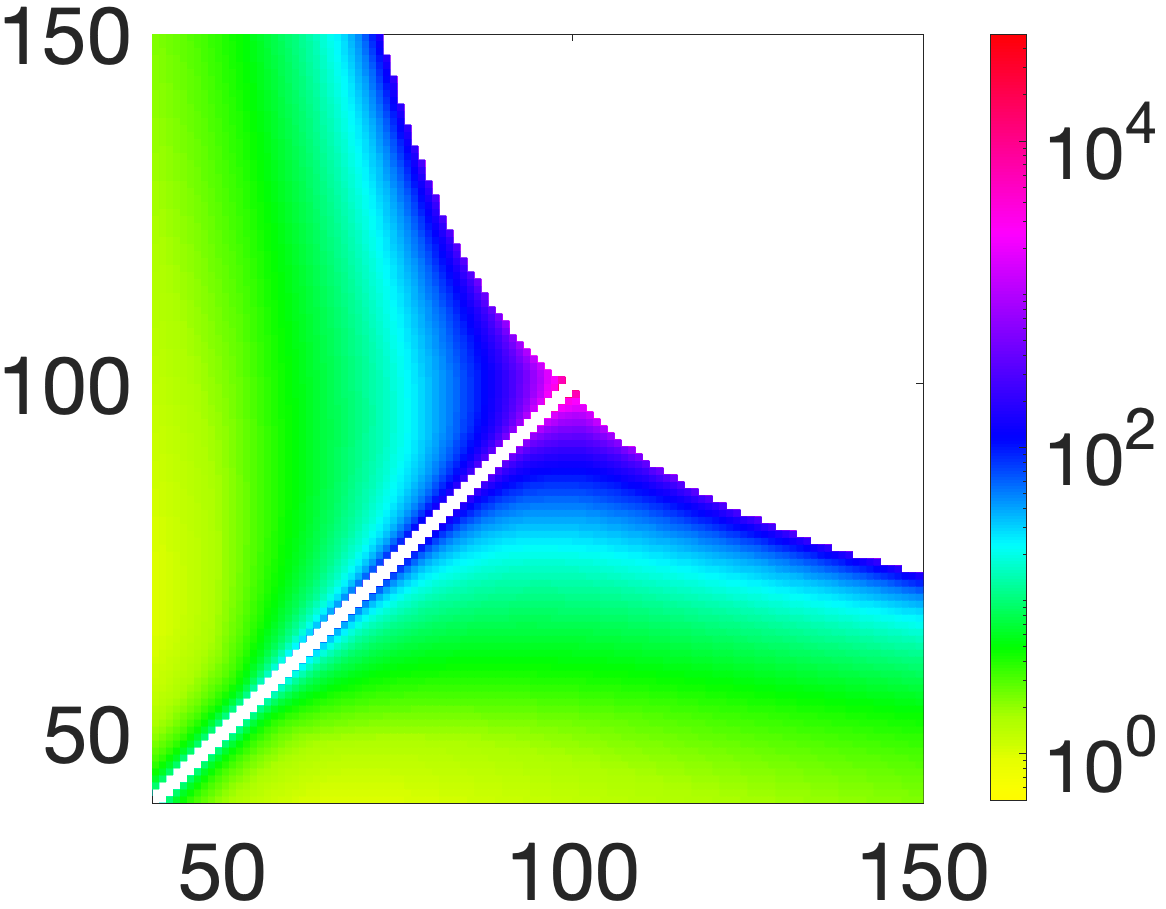}
    \caption{}
    \end{subfigure}%\hspace{-3em}
     \begin{subfigure}{0.22\textwidth}
    	\includegraphics[width=\textwidth]{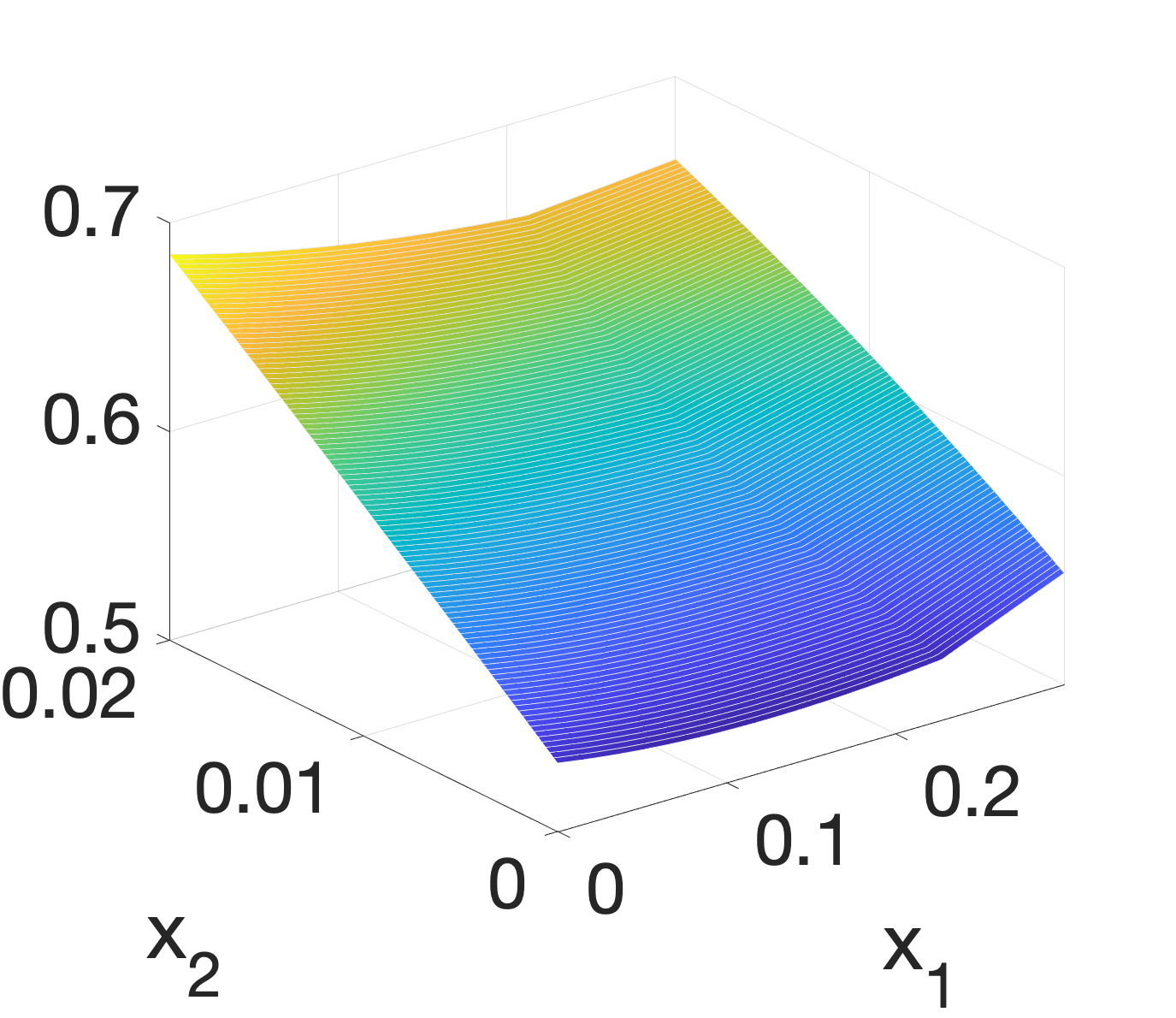}
    	\caption{}
     \end{subfigure}%\hspace{-1em}
          \begin{subfigure}{0.22\textwidth}
    	\includegraphics[width=\textwidth]{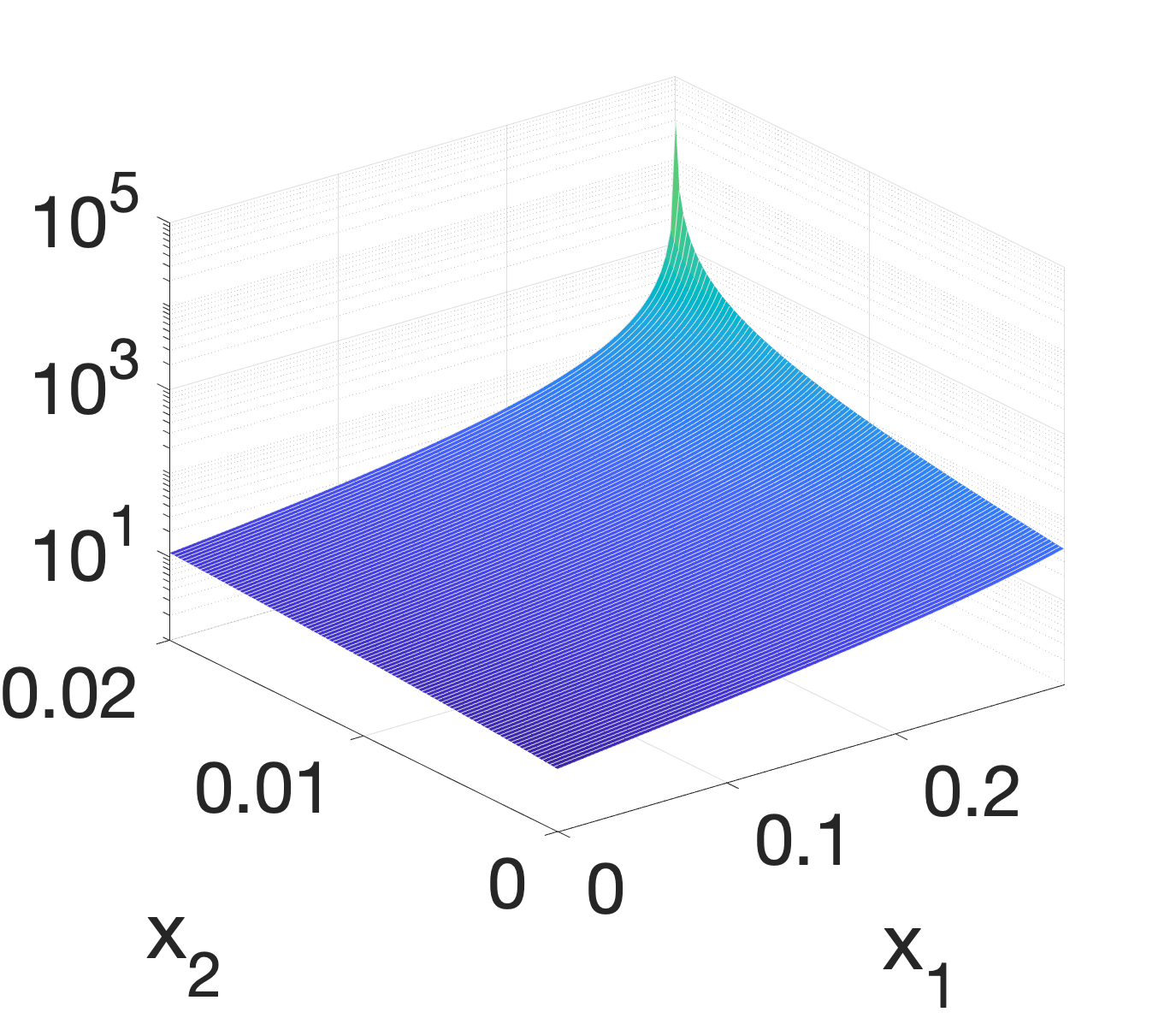}
    	\caption{}
     \end{subfigure}
    \caption{The plots of (a) $\hat{M}$ and (b) $\hat{M}_{est}$ as a function of tube potentials for the material pair (bone,Iodine) in the rectangle $\rR = [0,0.3]\times[0,0.02]$. The plots of $\vvvert J(x)^{-1}\vvvert$ on $\rR$ when (c) $tp = (40,74)$ and (d) $tp = (79,132)$.}
    \label{fig:BoneIodine}
\end{figure}

The above examples suggest that $\hat{M}$ and $\hat{M}_{est}$ are well-aligned in the case of the two-materials.

 \subsection{ME-CT: Three materials - three measurements}
 %%-------------------------------------------------------------------------
We now consider the materials bone, iodine, and water, in the said order. The considered rectangle is  $\rR = [0,0.3]\times[0,0.02]\times[0,3]$. In this case, we have
\begin{align}\label{hatM3d}
 \hat{M} = \max_{x \in \R^3} \vvvert \hJ(x)^{-1}\vvvert= \max_{z \in \rR} \max_{i \in \langle 3 \rangle} \{\vvvert J(z)_i^{-1}\vvvert,\vvvert J(z)^{-1}\vvvert\},
 \end{align}
 where $J(z)_i$ denotes the principal submatrix of the Jacobian $J(z)$ obtained by deleting the $i$-th row and column. Using positivity of the Jacobian matrix everywhere, one can estimate $\hat{M}$, according to theorem \ref{q-injectivity}, by
 \begin{align}\label{hatMest3d}
 \hat{M}_{est} = \frac{L^2}{\displaystyle \min_{i \in \langle 3 \rangle} \{L[J(z)]_i,\det J(z)\}}.
 \end{align}
 
  The plots for $\hat{M}$ and $\hat{M}_{est}$ are given in  fig. \ref{fig:BoneIodineWater}. The minimum value of $\hat{M}$, which is $23.49$, is attained at $(tp_1,tp_2,tp_3) = (40,74,150)$. However, it can also get very large, for example $\hat{M}$ equals 13235 at $(tp_1,tp_2,tp_3) = (70,120,150)$. 
  
  The minimum value for $\hat{M}_{est}$, which is equal to 1005, is attained at $(tp_1,tp_2,tp_3) = (40,67,150)$. We note that $\hat{M}(40,67,150) = 24.43$. We observed that although $\hat{M}$ and $\hat{M}_{est}$ display different overall behaviors, they attain their minimum at nearby points.
%%%%%%%%%%%%%%%%%%%%
   \begin{figure}[htbp]
   \centering
     \begin{subfigure}{0.24\textwidth}
    \includegraphics[width=\textwidth]{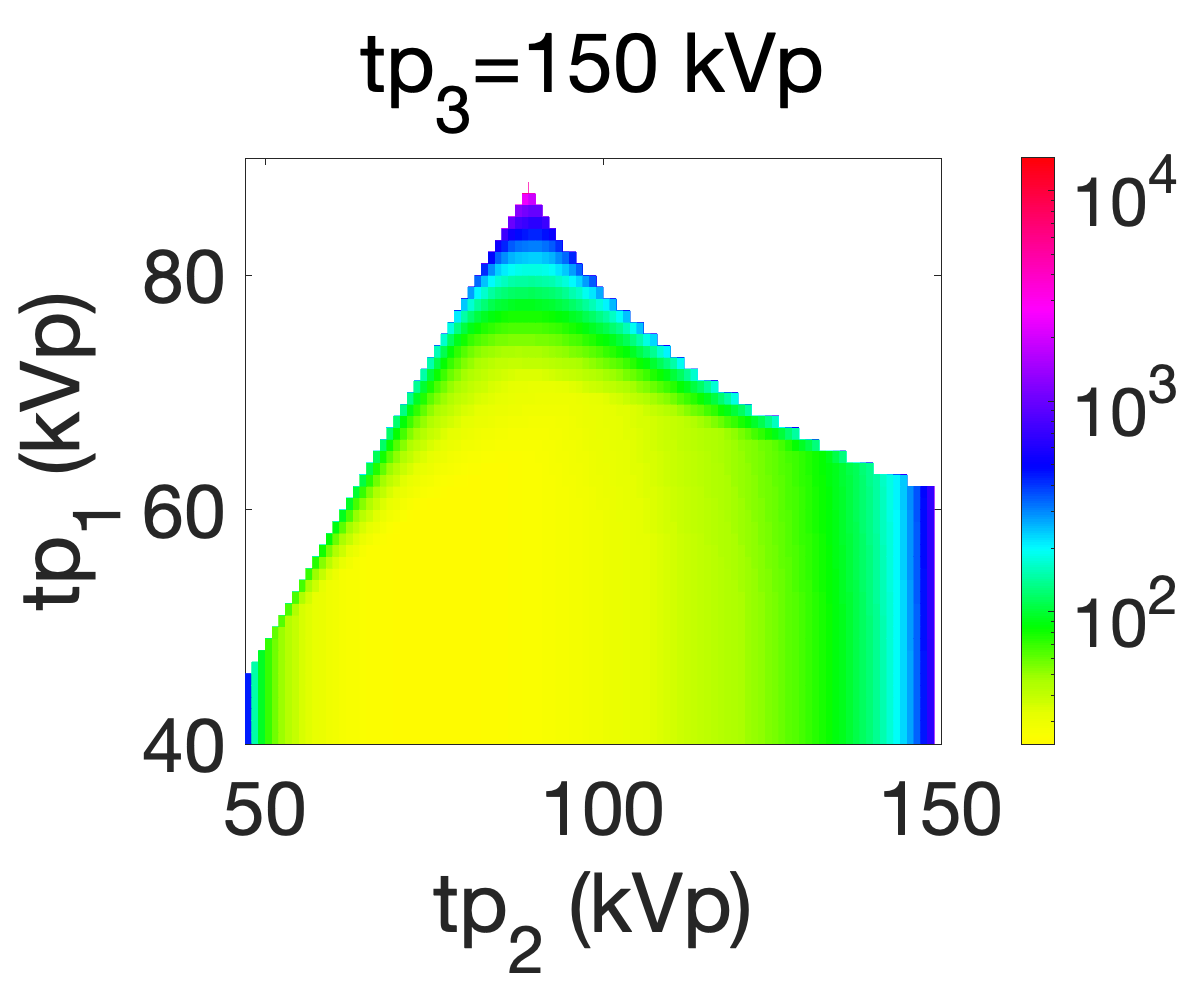}
    \caption{}
    \end{subfigure}%\hspace{-1em}
    \begin{subfigure}{0.24\textwidth}
    \includegraphics[width=\textwidth]{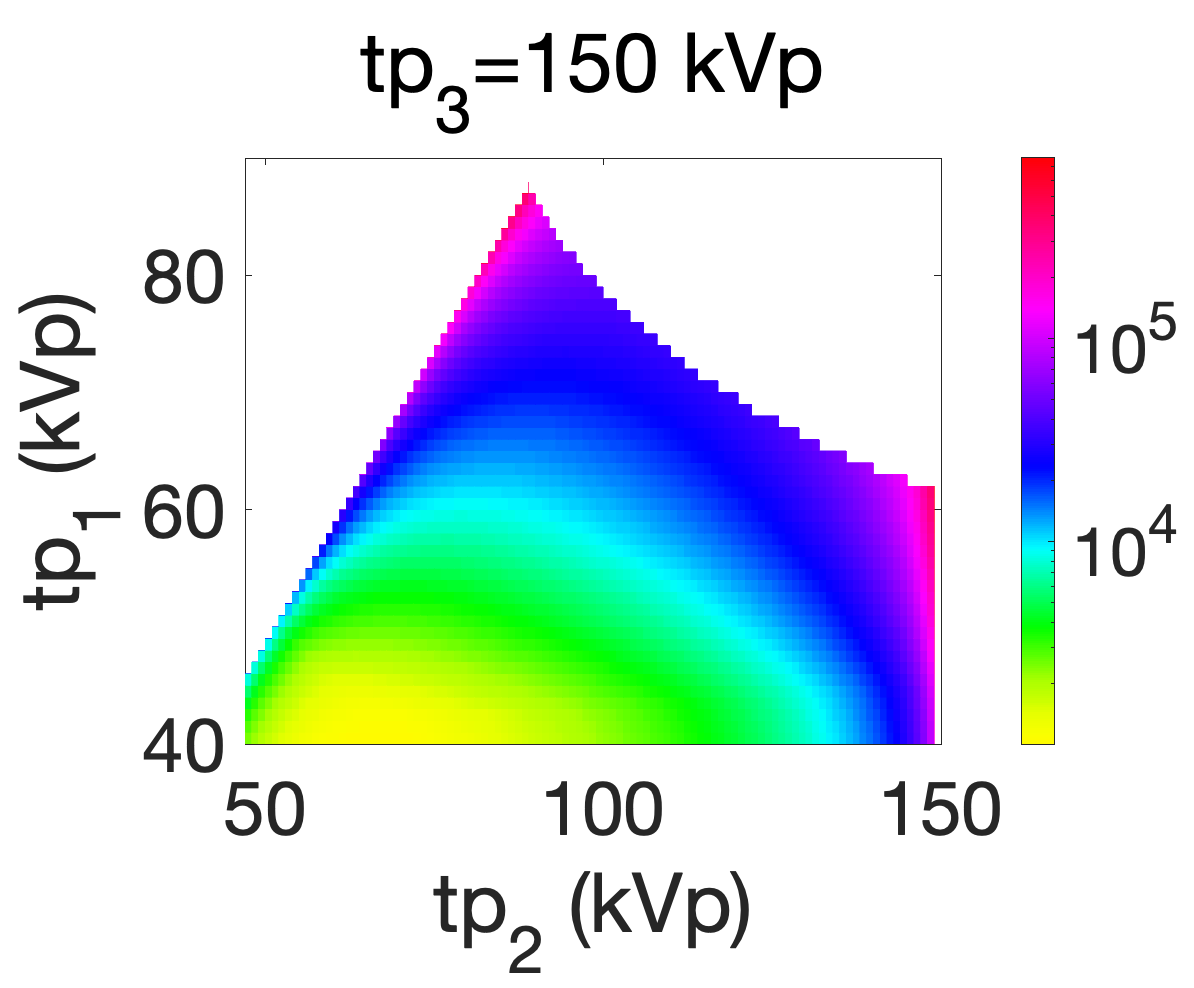}
    \caption{}
    \end{subfigure}
     \begin{subfigure}{0.24\textwidth}
    	\includegraphics[width=\textwidth]{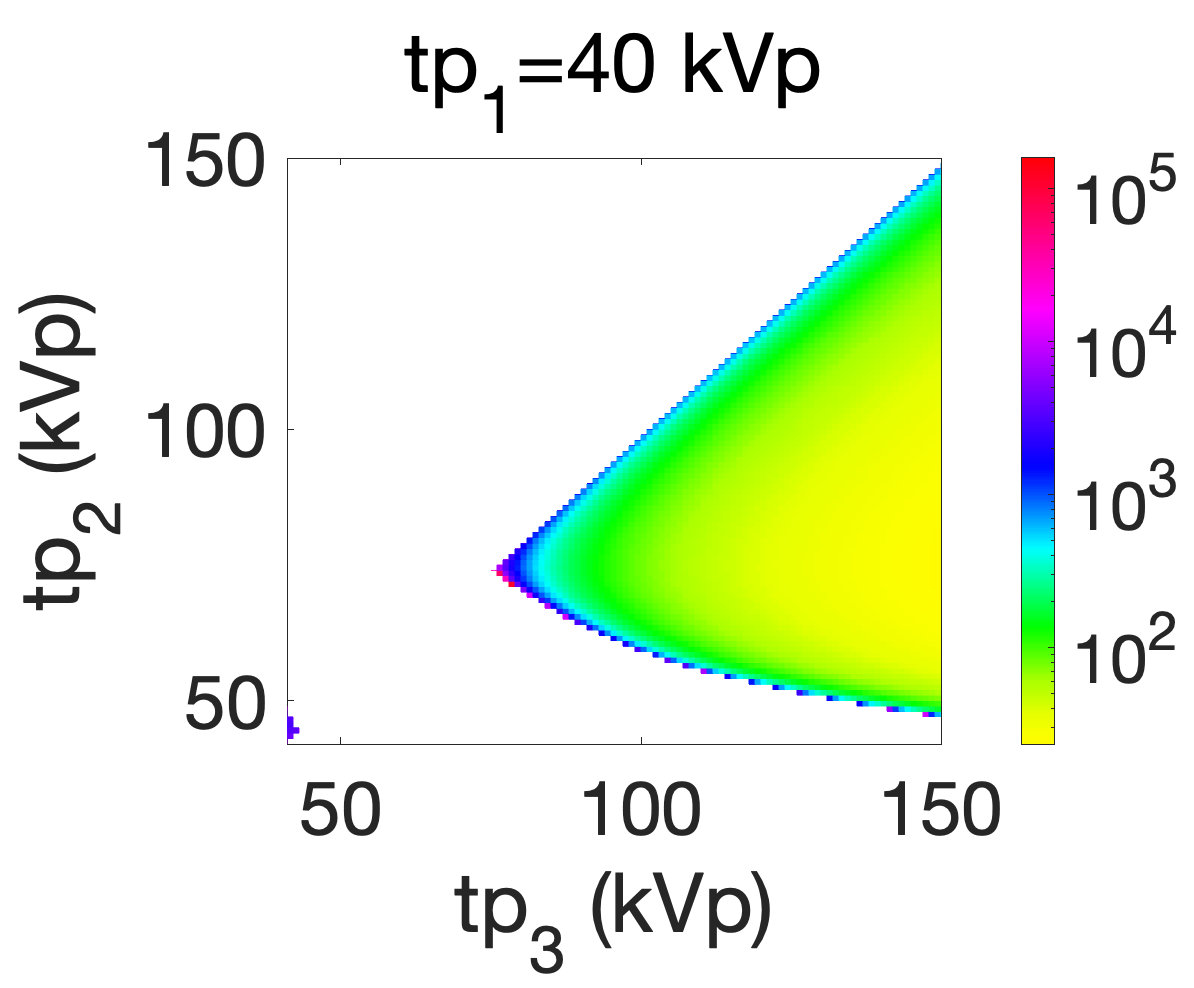}
    	\caption{}
    \end{subfigure}%\hspace{-1em}
     \begin{subfigure}{0.24\textwidth}
    	\includegraphics[width=\textwidth]{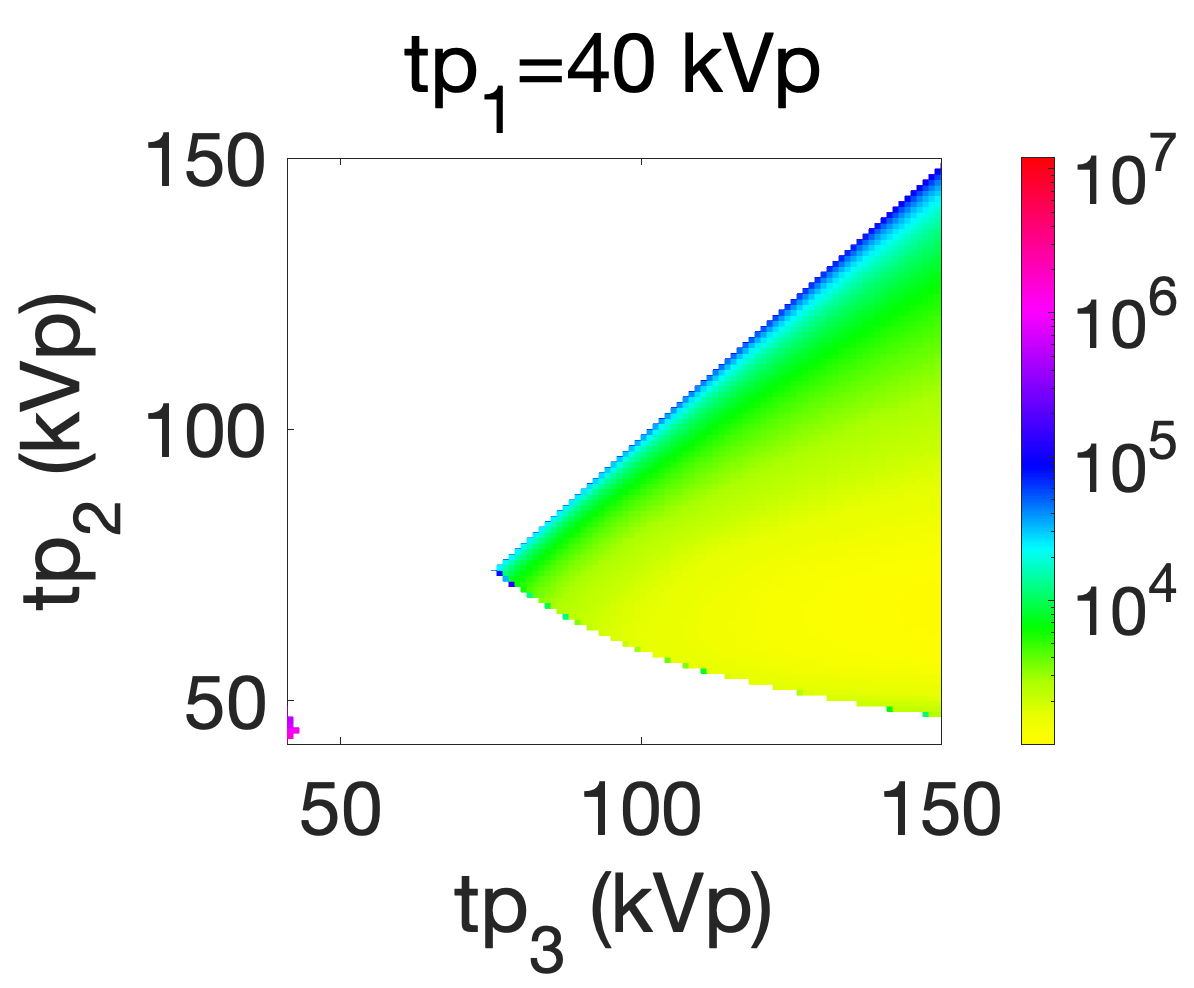}
    	\caption{}
     \end{subfigure}
    \caption{The plots of $\hat{M}$ and $\hat{M}_{est}$ as functions of tube potentials when (a,b) $tp_3=150$ kVp and when (c,d) $tp_1=40$ kVp for the material pair (bone,iodine,water) in the rectangle $\rR = [0,0.3]\times[0,0.02]\times[0,3]$.}
    \label{fig:BoneIodineWater}
\end{figure}

\subsection{Sinogram Reconstructions from DE-CT measurements}\label{s:reconstruction}
%%--------------------------------------------------------------------------------------------------------------------
In this section, we present the results of the numerical implementation of our inversion algorithm using MATLAB. We focus on reconstructing the line integrals (sinograms) from the DE-CT measurements. The second step, reconstructing material density maps from their sinograms can be obtained by standard Radon transform inversion, which is not presented here.

In the experiments, we considered a two-material (iodine and water) phantom supported in $[-2,2]\times[-2,2]$. The material density maps for iodine and water are given by 
$$\rho_1 = 0.05( \raisebox{2pt}{$\chi$}_{D_2} - \raisebox{2pt}{$\chi$}_{D_1})
\ \text{ and } \
\rho_2 = \raisebox{2pt}{$\chi$}_{D_3} - \raisebox{2pt}{$\chi$}_{D_2} + \raisebox{2pt}{$\chi$}_{D_1},$$
respectively. Here, $\raisebox{2pt}{$\chi$}_{D_i}, i=1,2,3,$ denotes the characteristic function of the disk centered at the origin and having radii 0.3, 0.5, and 1.5, respectively (see fig. \ref{fig:IodineWaterPhantomsDE-CTdata}(a)). 
 
The sinograms of the material densities $\rho_i, i=1,2,$ were analytically computed at 257 uniformly sampled nodes in $[-2\sqrt{2},2\sqrt{2}]$ and 400 uniformly sampled angles in $[-\pi,\pi]$ (see fig. \ref{fig:SinogramsReconstructions})(a).
 
Noise is considered in the measurements as follows. For each line $l$ and energy $E$, the number of measured photons in our model is $N_0S_i(E)e^{-M(E)\cdot x(l)}$ with $N_0$ the number of emitted photons independent of energy and line by normalization of our spectra. Poisson (shot) noise is then included in each such quantity.  The number $N_0$ of photons characterizing the Poisson distribution is chosen so that the relative $L_2$-errors between noisy and noiseless DE-CT measurements of $l\to I_i(l)$ with $tp = [40,68]$ for the low ($i=1$) and high ($i=2$) energies are approximately $0.6\%$ and $0.3\%$, respectively. When $tp = [55,82]$, the corresponding errors were approximately $0.4\%$ and $0.2\%$, respectively. In the experiments, we considered roughly $10^5$ emitted photons per energy bin (1 $keV$) per $mm^2$ (detector bin area) per $mAs$ (milliampere-seconds), and hence to a relatively noisy situation in practice \cite{Lasio}. This value was chosen to display small but sizeable errors for the stable tube profile. They generated very large errors for the unstable tube profile. 

The plots of DE-CT measurements corresponding to the two pairs of tube potentials are shown in fig. \ref{fig:IodineWaterPhantomsDE-CTdata}(b,c)). 
 
We tested our inversion algorithm on noisy DE-CT measurements by considering both the original DE-CT map and its extension. Reconstructions of the sinograms, obtained by using the extended map, from noisy DE-CT measurements for the two pairs of tube potentials are shown in fig. \ref{fig:SinogramsReconstructions}(b,c). Fig. \ref{fig:ProfilesSinogramsReconstructions} contains the central horizontal slices of the reconstructions shown in fig. \ref{fig:SinogramsReconstructions}(b,c). The relative $L_2$-errors between exact sinograms (denoted by $X_1$ and $X_2$) and their reconstructions are given in table \ref{tab:L2errors}. As we expected, using the tube potentials minimizing the inverse Jacobian led to considerably better reconstructions. In all cases, using the extended map instead of the original enhanced the quality of the reconstructions. When we considered the DE-CT measurements corresponding to the tube potentials $tp = (55, 82)$, and performed inversion using the original map, the algorithm failed to converge for approximately 17\% of the lines considered. The relative $L_2$-errors computed by excluding these lines is given in the last column of table \ref{tab:L2errors}.

\begin{figure}[htbp]
\centering
\begin{subfigure}{0.19\textwidth}
    \includegraphics[width=\textwidth]{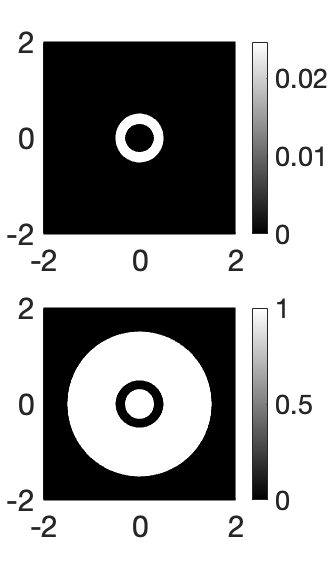}
    \caption{}
\end{subfigure}
\begin{subfigure}{0.32\textwidth}
    \includegraphics[width=\textwidth]{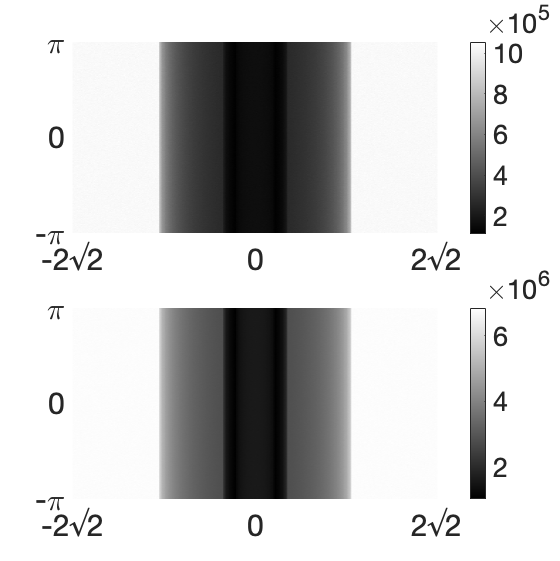}
    \caption{}
\end{subfigure}
\begin{subfigure}{0.32\textwidth}
    \includegraphics[width=\textwidth]{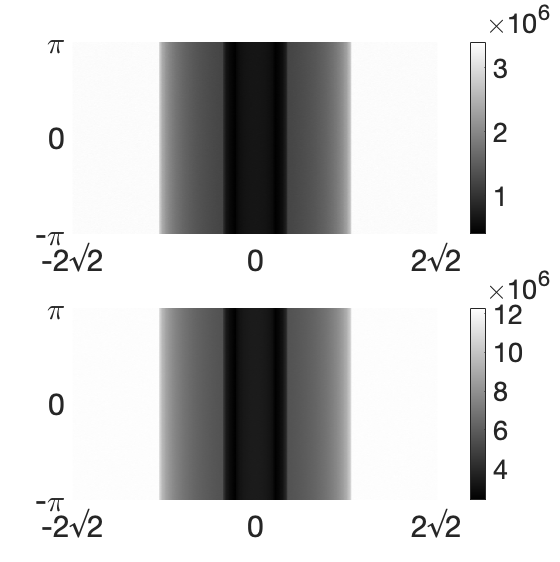}
    \caption{}
\end{subfigure}
    \caption{(a)Iodine (top) and water (bottom) density plots of the two-material phantom, DE-CT measurements corresponding to tube potentials (b) $tp = (40,68)$ and (c) $tp = (55,82)$.}
    \label{fig:IodineWaterPhantomsDE-CTdata}
\end{figure}

\begin{figure}[htbp]
\centering
\begin{subfigure}{0.32\textwidth}
    \includegraphics[width=\textwidth]{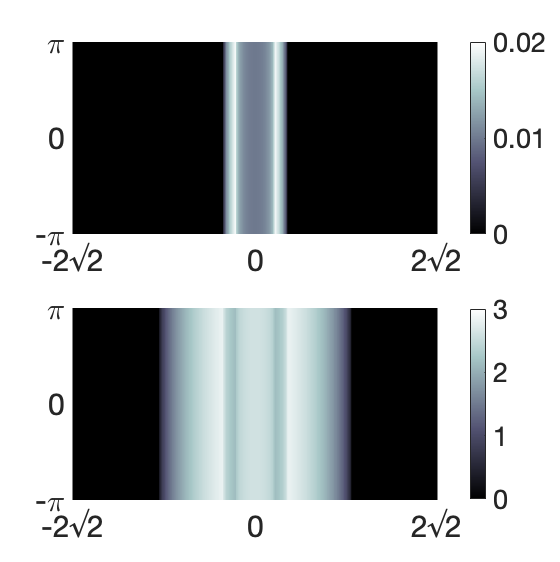}
    \caption{}
\end{subfigure}
\begin{subfigure}{0.32\textwidth}
    \includegraphics[width=\textwidth]{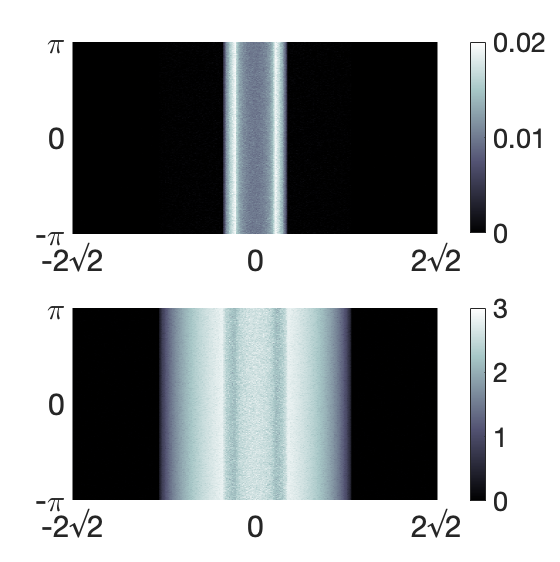}
    \caption{}
\end{subfigure}
\begin{subfigure}{0.32\textwidth}
    \includegraphics[width=\textwidth]{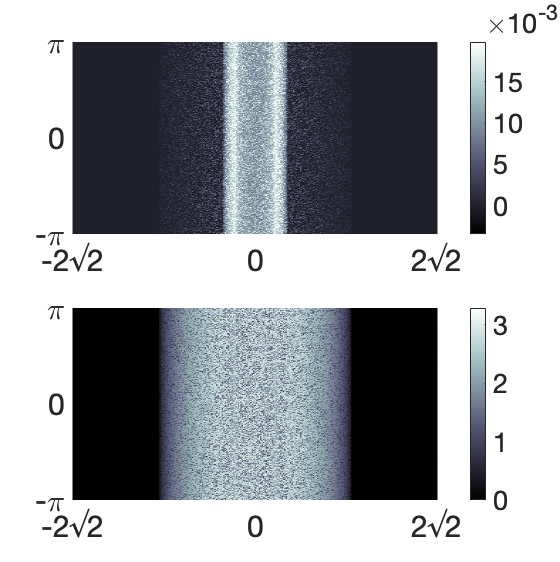}
    \caption{}
\end{subfigure}
    \caption{(a) The exact sinograms of iodine (top) and water (bottom) density maps of the phantom given in Fig. \ref{fig:IodineWaterPhantomsDE-CTdata}, and their reconstructions from noisy DE-CT measurements corresponding to tube potentials (b) $tp = (40,68)$ and (c) $tp = (55,82)$.}
    \label{fig:SinogramsReconstructions}
\end{figure}

\begin{figure}[htbp]
\centering
\begin{subfigure}{0.4\textwidth}
    \includegraphics[width=\textwidth]{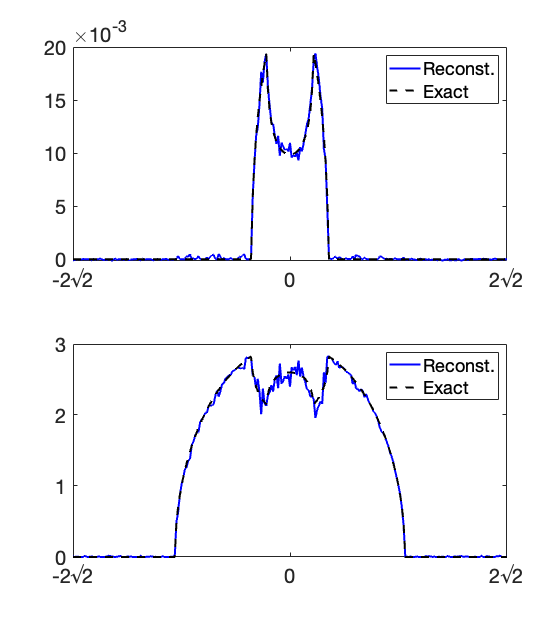}
    \caption{}
\end{subfigure}
\begin{subfigure}{0.4\textwidth}
    \includegraphics[width=\textwidth]{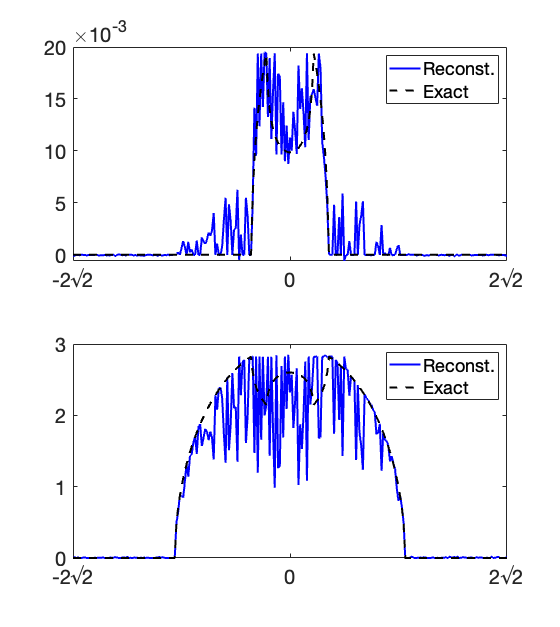}
    \caption{}
\end{subfigure}
    \caption{Comparison of central profiles of the exact sinograms and their reconstructions from noisy DE-CT measurements corresponding to tube potentials (a) $tp = (40,68)$ and (b) $tp = (55,82)$.}
    \label{fig:ProfilesSinogramsReconstructions}
\end{figure}

\begin{table}[htbp]
\begin{tabular}{c|c|c|c|c|}
\cline{2-5}
                              & \multicolumn{2}{c|}{Extension} & \multicolumn{2}{c|}{Original map} \\ \cline{2-5} 
                              & $tp = (40,68)$  & $tp = (55,82)$  & $tp = (40,68)$ & $tp = (55,82)$       \\ \hline
\multicolumn{1}{|c|}{$L_2$-error for $X_1$} & 0.0483      & 0.3408    & 0.0591         & 0.4792$^*$        \\ \hline
\multicolumn{1}{|c|}{$L_2$-error for $X_2$} & 0.0312      & 0.2280     & 0.0380         & 0.3105$^*$       \\ \hline
\end{tabular}
\caption{Normalized $L_2$-errors of sinogram reconstructions using our inversion algorithm. $^*$Computed excluding the lines for which the algorithm failed to converge and hence shown for illustrative purposes.}
     \label{tab:L2errors}
\end{table}
%--------------------
\section{Conclusions}
%--------------------
Finding local criteria for the global injectivity of maps from a subset of $\R^n$ to $\R^n$ is a challenging problem. One such criterion is based on the notion of $P-$functions we considered in this paper. Following \cite{BalTer20}, many Multi-Energy Computed Tomography (ME-CT) problems were shown to satisfy the latter criterion. 

Following \cite{MasColell}, we obtained in section 2 several extensions of a given $P-$function to possibly smooth injective maps from $\R^n$ to $\R^n$ and presented explicit stability estimates controlling the propagation of measurement errors to reconstruction errors. 

While stability is guaranteed for extended $P-$functions, standard algorithms such as those based on Newton or Gauss-Seidel methods are not guaranteed to converge to the unique (fixed point) solution. We propose in section 3 an algorithm, taking the form of a damped Newton method, that is guaranteed to converge to the fixed point, first with linear rate of convergence, and then with quadratic rate of convergence using a standard Newton method when the discrete dynamics are sufficiently close to the fixed point. The algorithm was shown to converge for (sufficiently) smooth extensions of the original problem. The algorithm also most likely converges (with $h$ sufficiently small) for the piecewise-smooth Mas-Colell extension \eqref{Fext} but we do not have a complete proof in that case. 

However, we showed that choosing sufficiently small values of the damping parameter $h$ was necessary to obtain a convergent algorithm. We presented examples of non-convergent discrete cyclic trajectories for smooth $P-$functions with Jacobians with positive entries (as is the case in ME-CT applications) when $h=1$ (standard Newton algorithm) as well as $h$ close to $1$.

We finally considered ME-CT inversions in section 4. We focused on the first step, namely the reconstruction of line integrals of absorption densities from ME-CT measurements. The second step, reconstructing spatially varying functions from their line integrals amounts to a standard inverse Radon transform procedure, which is not presented here. We considered two-  ($n=2$) and three- ($n=3$) dimensional settings. We showed that the errors in the reconstructions strongly depended on the choice of energy measurements (e.g., tube potentials). While not perfect, we also showed that the selection of tube potentials based on the stability estimates of section 2, as opposed to the numerical evaluation of the inverse Jacobian, also proved reasonable. 

Numerical reconstructions of two-material line integrals (sinograms) from DE-CT measurements using the proposed damped Newton algorithm confirmed the theoretical predictions. Reconstructions were shown to be significantly more stable for optimized choices of the tube potentials. Moreover, for less stable tube potentials, we observed that using the extended map significantly improved the reconstructions, with the damped Newton algorithm failing to converge in unstable cases and in the presence of sufficiently large noise.

\section*{Acknowledgment}
The authors thank Emil Sidky for useful discussions and references. This research was partially supported by the National Science Foundation, Grants DMS-1908736 and EFMA-1641100 and by the Office of Naval Research, Grant N00014-17-1-2096. 

%%----------------------------------------------
\begin{appendix}
%%----------------------------------------------
\section{}
In this section, we provide two alternatives to the constant $\tau$ given in \eqref{tau}.

\subsection{An alternative estimate}
If $A$ is a $P-$matrix, then the set of positive eigenvalues of all principal submatrices of $A$, denoted by $\Lambda_A$, contains at least the diagonal entries of $A$, and hence it is nonempty. The constant 
$$\mu_A := \min(\Lambda_A),$$
 can be seen as a characteristic quantity for $P-$matrices.

\begin{prop}[\cite{BalTer20}]\label{mu_A}
If $A$ is a $P-$matrix, then $A-\lambda \I$ is a $P-$matrix for all $0 \leq \lambda < \mu = \mu_A$, and $A-\mu \I$ is a $P_0-$matrix.
\end{prop}

The constant $\mu$ is particularly useful in obtaining a lower bound for the determinant of a $P-$matrix.
\begin{prop}\label{t:detP-estimate}
Let $A$ be a $n \times n$ $P-$matrix with $k$, $0\leq k \leq n$, real eigenvalues and $\mu = \mu_A$. Then, 
\begin{align}\label{detP-estimate}
\det A \geq \left(\sin \frac{\pi}{n}\right)^{n-k} \mu^n.
\end{align} 
\end{prop}
\begin{proof}
Since $A$ is a $P-$matrix, any real, hence positive, eigenvalue of $A$ is bounded below by $\mu$. Also, the eigenvalues of $A$ are given by $\mu + \lambda$ with $\lambda$ being an eigenvalue of $B = A-\mu \I$ which is a $P_0-$matrix. By Kellogg's theorem \cite{Kellogg}, $| \arg \lambda| \leq \pi(1 - \tfrac{1}{n})$, and thus $|\mu + \lambda| \geq \mu \sin \frac{\pi}{n}$. Since the determinant of a matrix is equal to the product of its eigenvalues, we obtain \eqref{detP-estimate}.
\end{proof}

\begin{prop}\label{mu_hatJ}
Let $\hJ_\e$ be given as in \eqref{Jhat_e}. Then, for all $\e \geq 0$ and $x \in \R^n$, we have 
$$\mu_{\hJ_\e(x)} \geq  \mu_{J(P_\e(x))},$$ 
with equality if $\e = 0$. Thus, in view of proposition \ref{t:detP-estimate}, we have
\begin{align}\label{detJhat-estimate}
\det \hJ_\e(x) \geq \left(\sin \frac{\pi}{n}\right)^{n-1} \mu_{J(P_\e(x))}^n,
\end{align} 
and
\begin{align}\label{invJnorm_e-alt}
\vvvert \hJ_\e(x)^{-1}\vvvert \leq \frac{ (n-1)^{\frac{n-1}{2}} L_\e^{n-1}}{\left(\sin \frac{\pi}{n}\right)^{n-1} \mu_{J(P_\e(x))}^n}.
\end{align}
\end{prop}
\begin{proof}
Observe that
$
\hJ_\e(x)-\lambda \I = (J(P_\e(x))-\lambda \I)DP_\e(x) +(L_\e-\lambda)(\I-DP_\e(x)).
$
Then, using the properties of the determinant, we obtain
\begin{align}\label{char_poly}
\det(\hJ_\e(x)-\lambda \I) = \sum_{K \subset \langle n\rangle} c_K(x) (L_\e-\lambda)^{|K|}[J(P_\e(x))-\lambda \I]_K,
\end{align}
where 
\begin{align*}
c_K(x) = \Big(\prod_{k \in K} (1-p'_\e(x_k))\Big) \Big(\prod_{k \in \langle n\rangle \setminus K}p'_\e(x_k) \Big).
\end{align*}
We note that, since $0 \leq p'_\e \leq 1$, we have $0 \leq c_K(x) \leq 1$ for all $x \in \R^n$ and $K \subset \langle n\rangle$. Moreover, for all $x \in \R^n$, there is $K \subset \langle n\rangle$ such that $c_K(x) > 0$. 

Now let $x \in \R^n$ be arbitrary, and suppose that $\lambda$ is an eigenvalue of $\hJ_\e(x)$. If $\lambda < \mu_{J(P_\e(x))}$, then $[J(P_\e(x))-\lambda \I]_K >0$ for all $K\subset \langle n\rangle$ by proposition \ref{mu_A}. Also, since $\mu_{J(P_\e(x))} \leq J(P_\e(x))_{ii} \leq L_\e$, we obtain that the right hand side of \eqref{char_poly} is positive while the left hand side is zero, which is a contradiction. Hence, we must have $\lambda \geq \mu_{J(P_\e(x))}$. Applying the same argument to the eigenvalues of the principal submatrices of $\hJ_\e(x)$, we obtain that $ \mu_{\hJ_\e(x)} \geq \mu_{J(P_\e(x))}.$ 

For $\e = 0$, we observe that $c_K = 0$ for all but one $K' \subset \langle n \rangle$, for which $c_{K'} = 1$. This implies that the eigenvalues of submatrices of $\hJ_\e(x)$ are either equal $L_\e$ or coincide with an eigenvalue of a submatrix of $J(P_\e(x))$. Thus, $ \mu_{\hJ(x)} = \mu_{J(P(x))}.$

Then, the estimate \eqref{detJhat-estimate} follows from proposition \ref{t:detP-estimate}. Finally, proceeding as in the proof of theorem \ref{t:invF is Lipschitz} and using \eqref{detJhat-estimate} in estimating the determinant of $\hJ_\e$, we obtain \eqref{invJnorm_e-alt}.
\end{proof}

\begin{definition}
  Let $F$ be a continuously differentiable map on a closed rectangle $\rR$ with a $P-$matrix Jacobian $J(x)$ at every $x \in \rR$. The quantity
\begin{align}\label{mu}
		\mu_F &:= \min_{x \in \rR} \mu_{J(x)},
\end{align}
is called the injectivity constant of $F$.
\end{definition}

\begin{theorem}\label{t:invF is Lipschitz-alt}
Let $\hF$ be given as in \eqref{Fext}. Then, for all $x, y \in \R^n$,   
\begin{align}\label{invFhat is Lipschitz-alt}
 \|\hF(x)-\hF(y)\|_\infty \geq \gamma \|x-y\|_\infty,
\end{align}
where
\begin{align}\label{gamma}
   \gamma := \frac{ \left( \sin \frac{\pi}{n}\right)^{n-1}}{n^{\frac{1}{2}}(n-1)^{\frac{n-1}{2}} } \frac{\mu_F^n}{L^{n-1}}.
\end{align}
\end{theorem}
\begin{proof}
Proceeding as in the proof of theorem \ref{t:invF is Lipschitz}, using \eqref{invJnorm_e-alt} in estimating the determinant of $\vvvert \hJ_\e(x)^{-1}\vvvert$, and finally letting $\e \to 0$, we obtain \eqref{invFhat is Lipschitz-alt}.
\end{proof}

The comparison of the estimates \eqref{invFh is Lipschitz} and \eqref{invFhat is Lipschitz-alt} is given in the following proposition.
\begin{prop}
Let $A$ be a $n \times n$ $P-$matrix. Then, for any $K \subset \langle n \rangle$, we have
\begin{align}\label{tau_vs_gamma}
 \vvvert A\vvvert^{|K|} [A]_K \geq \left(\sin \frac{\pi}{n}\right)^{n-1} \mu_A^n,
\end{align} 
and thus $\tau \geq \gamma$ where $\tau$ and $\gamma$ are given in \eqref{tau} and \eqref{gamma}, respectively.
\end{prop}
\begin{proof}
Let $K \subset \langle n \rangle$ be arbitrary. By definition of $\mu_A$, we have $\vvvert A\vvvert \geq a_{ii} \geq \mu_A$ for all $i \in \langle n \rangle$, which implies the result if $K = \langle n \rangle$. For otherwise,  we apply proposition \ref{t:detP-estimate} to $A_K$ to obtain
$$[A]_K \geq \left(\sin \frac{\pi}{n-|K|}\right)^{n-|K|-1} \mu_{A_K}^{n-|K|} \geq \left(\sin \frac{\pi}{n}\right)^{n-1} \mu_A^{n-|K|}.$$
Therefore, 
$$ \vvvert A\vvvert^{|K|} [A]_K \geq \mu_A^{|K|}\left(\sin \frac{\pi}{n}\right)^{n-1} \mu_A^{n-|K|} = \left(\sin \frac{\pi}{n}\right)^{n-1} \mu_A^n,$$
which can be used to obtain the inequality $\tau \geq \gamma$.
\end{proof}

%%%%%%%%%%%%%%%%%%%%%%%%%%%%%%%%%%%%
\subsection{An estimate without using an extension}
In this section, we derive an estimate that does not require any extension. We start with the following geometric property of $P-$matrices. 
\begin{theorem}[\cite{FiedlerPtak,GaleNikaido}]\label{t:GeoCharP}
An $n \times n$ matrix $A$ is a $P-$matrix if and only if $A$ reverses the sign of no vector except zero, that is for every nonzero vector $v \in \R^n$, there is an index $i \in \langle n \rangle$ such that $v_i(Av)_i> 0$.
\end{theorem}

Theorem \ref{t:GeoCharP} was used in \cite{MathiasPang} to obtain another characteristic quantity for $P-$matrices. Evidently, $A$ is a $P-$matrix if and only if
\begin{align}\label{alphaA}
\alpha_A := \min_{\|v\|_\infty = 1} \max_{i \in \langle n \rangle} v_i(Av)_i>0.
\end{align}
Consequently, for all $v \in \R^n$ (see also \cite[Lemma 3.12]{MoreRheinboldt}),
\begin{align}
\max_{i \in \langle n \rangle} v_i(Av)_i \geq \alpha_A \|v\|^2_\infty. 
\end{align}

We note that $\alpha_A \leq \mu_A$. Indeed, since $A-\mu_A \I$ is no longer a $P-$matrix, we must have $0 \geq \alpha_{(A-\mu_A \I)} \geq  \alpha_A - \mu_A $. We can now obtain the following quantitative estimate of injectivity.
\begin{theorem}\label{t:invI is Lipschitz}
  Let $\rR \subset \R^n$ be a closed rectangle. Suppose that $F: \rR \to \R^n$ is a continuously differentiable map with a $P-$matrix Jacobian $J(z)$ at every $z \in \rR$. We define
 \begin{align}\label{alpha}
	\alpha := \min_{z\in \rR} \alpha_{J(z)}.
\end{align}
   Then, for al $x, y \in \rR$,
\begin{align}\label{alpha-estimate}
\|F(x)-F(y)\|_\infty  \geq \alpha \|x-y\|_\infty.
\end{align}
\end{theorem}
\begin{proof}
Let $x, y \in \rR$. If $x=y$, we are done, so we assume that $x \neq y$. For each $i  \in \langle n \rangle$, we define
$$g_i: [0,1] \to \R, \; g_i(t) = F_i(tx + (1-t)y).$$
Then, by the Mean Value Theorem, there exists $t_i \in (0,1)$ such that $g_i(1) - g_i(0) = g_i'(t_i)$. Observing that $g_i(1) = F_i(x)$, $g_i(0) = F_i(y)$, and $\textstyle g_i'(t) = \sum_{j=1}^n \frac{\partial F_i}{\partial x_j}((tx + (1-t)y) (x_j-y_j)$, we obtain
\begin{align}
F_i(x)-F_i(y) = (J(z_i) (x-y))_i,
\end{align}
where $z_i = t_ix+(1-t_i)y$ for some $t_i \in (0,1)$. 

Since the Jacobian $J(z)$ is a $P-$matrix at every $z \in \rR$, we have
\begin{align}
\max_{ i \in \langle n \rangle} v_i(J(z) v)_i \geq \alpha_{J(z)} \|v\|^2_\infty \geq \alpha \|v\|^2_\infty,
\end{align}
for all $v \in \R^n$. Thus, we obtain
\begin{align*}
\|x-y\|_\infty \|F(x)-F(y)\|_\infty 
&\geq \max_{i  \in \langle n \rangle} (x_i-y_i) (F_i(x)-F_i(y)) \\
&= \max_{i  \in \langle n \rangle} (x_i-y_i) (J(z_i) (x-y))_i\
 \geq \ \alpha  \|x-y\|^2_\infty.
\end{align*}

Finally, since $x \neq y$, we can divide both sides by $\|x-y\|_\infty$ and obtain \eqref{alpha-estimate}.
\end{proof}

\end{appendix}

%%%%%%%%%%%%
\bibliographystyle{ieeetr}
\bibliography{P-functionInversion}

\begin{thebibliography}{10}

\bibitem{AlvarezMacovski}
R.~E. Alvarez and A.~Macovski, ``Energy-selective reconstructions in x-ray
  computerised tomography,'' {\em Physics in Medicine \& Biology}, vol.~21,
  no.~5, p.~733, 1976.

\bibitem{Lionheart}
W.~R. Lionheart, B.~T. Hjertaker, R.~Maad, I.~Meric, S.~B. Coban, and G.~A.
  Johansen, ``Non-linearity in monochromatic transmission tomography,'' {\em
  arXiv preprint arXiv:1705.05160}, 2017.

\bibitem{Katsura}
M.~Katsura, J.~Sato, M.~Akahane, A.~Kunimatsu, and O.~Abe, ``Current and novel
  techniques for metal artifact reduction at {CT}: practical guide for
  radiologists,'' {\em Radiographics}, vol.~38, no.~2, pp.~450--461, 2018.

\bibitem{McCollough}
C.~H. McCollough, S.~Leng, L.~Yu, and J.~G. Fletcher, ``Dual-and multi-energy
  {CT}: principles, technical approaches, and clinical applications,'' {\em
  Radiology}, vol.~276, no.~3, pp.~637--653, 2015.

\bibitem{Park}
H.~S. Park, Y.~E. Chung, and J.~K. Seo, ``Computed tomographic beam-hardening
  artefacts: mathematical characterization and analysis,'' {\em Philosophical
  Transactions of the Royal Society A: Mathematical, Physical and Engineering
  Sciences}, vol.~373, no.~2043, p.~20140388, 2015.

\bibitem{Schlomka}
J.~Schlomka, E.~Roessl, R.~Dorscheid, S.~Dill, G.~Martens, T.~Istel,
  C.~B{\"a}umer, C.~Herrmann, R.~Steadman, G.~Zeitler, {\em et~al.},
  ``Experimental feasibility of multi-energy photon-counting {K}-edge imaging
  in pre-clinical computed tomography,'' {\em Physics in Medicine \& Biology},
  vol.~53, no.~15, p.~4031, 2008.

\bibitem{Taguchi}
K.~Taguchi, ``Energy-sensitive photon counting detector-based {X}-ray computed
  tomography,'' {\em Radiological physics and technology}, vol.~10, no.~1,
  pp.~8--22, 2017.

\bibitem{Willemink}
M.~J. Willemink, M.~Persson, A.~Pourmorteza, N.~J. Pelc, and D.~Fleischmann,
  ``Photon-counting {CT}: technical principles and clinical prospects,'' {\em
  Radiology}, vol.~289, no.~2, pp.~293--312, 2018.

\bibitem{Heismann2012}
B.~J. Heismann, B.~T. Schmidt, and T.~Flohr, ``Spectral computed tomography,''
  SPIE Bellingham, WA, 2012.

\bibitem{So2021}
A.~So and S.~Nicolaou, ``Spectral computed tomography: Fundamental principles
  and recent developments,'' {\em Korean Journal of Radiology}, vol.~22, no.~1,
  p.~86, 2021.

\bibitem{Spektr3}
J.~Punnoose, J.~Xu, A.~Sisniega, W.~Zbijewski, and J.~H. Siewerdsen,
  ``Technical note: spektr 3.0 - {A} computational tool for x-ray spectrum
  modeling and analysis,'' {\em Medical Physics}, vol.~43, no.~8Part1,
  pp.~4711--4717, 2016.

\bibitem{NIST}
J.~H. Hubbell and S.~M. Seltzer, ``Tables of x-ray mass attenuation
  coefficients and mass energy-absorption coefficients 1 ke{V} to 20 {M}e{V}
  for elements {Z}= 1 to 92 and 48 additional substances of dosimetric
  interest,'' tech. rep., National Inst. of Standards and Technology-PL,
  Gaithersburg, MD (United~States), 1995.

\bibitem{Alvarez19}
R.~E. Alvarez, ``Invertibility of the dual energy x-ray data transform,'' {\em
  Medical Physics}, vol.~46, no.~1, pp.~93--103, 2019.

\bibitem{BalTer20}
G.~Bal and F.~Terzioglu, ``Uniqueness criteria in multi-energy {CT},'' {\em
  Inverse Problems}, vol.~36, no.~6, p.~065006, 2020.

\bibitem{GaleNikaido}
D.~Gale and H.~Nikaido, ``The {J}acobian matrix and global univalence of
  mappings,'' {\em Mathematische Annalen}, vol.~159, pp.~81--93, Apr 1965.

\bibitem{FiedlerPtak}
M.~Fiedler and V.~Ptak, ``On matrices with non-positive off-diagonal elements
  and positive principal minors,'' {\em Czechoslovak Mathematical Journal},
  vol.~12, no.~3, pp.~382--400, 1962.

\bibitem{MoreRheinboldt}
J.~Mor{\'e} and W.~Rheinboldt, ``On {P}-and {S}-functions and related classes
  of n-dimensional nonlinear mappings,'' {\em Linear Algebra and its
  Applications}, vol.~6, pp.~45--68, 1973.

\bibitem{More1972}
J.~J. Mor{\'e}, ``Nonlinear generalizations of matrix diagonal dominance with
  application to {G}auss--{S}eidel iterations,'' {\em SIAM Journal on Numerical
  Analysis}, vol.~9, no.~2, pp.~357--378, 1972.

\bibitem{MasColell}
A.~Mas-Colell, ``Homeomorphisms of compact, convex sets and the {J}acobian
  matrix,'' {\em SIAM Journal on Mathematical Analysis}, vol.~10, no.~6,
  pp.~1105--1109, 1979.

\bibitem{Gradshteyn}
I.~S. Gradshteyn and I.~M. Ryzhik, {\em Table of Integrals, Series, and
  Products}.
\newblock Elsevier, 2007.

\bibitem{MarGorZam94}
G.~De~Marco, G.~Gorni, and G.~Zampieri, ``Global inversion of functions: an
  introduction,'' {\em Nonlinear Differential Equations and Applications
  NoDEA}, vol.~1, no.~3, pp.~229--248, 1994.

\bibitem{Rhe78}
W.~Rheinboldt, ``An adaptive continuation process for solving systems of
  nonlinear equations,'' {\em Banach Center Publications}, vol.~3, no.~1,
  pp.~129--142, 1978.

\bibitem{Lasio}
G.~M. Lasio, B.~R. Whiting, and J.~F. Williamson, ``Statistical reconstruction
  for x-ray computed tomography using energy-integrating detectors,'' {\em
  Physics in medicine \& biology}, vol.~52, no.~8, p.~2247, 2007.

\bibitem{Kellogg}
R.~Kellogg, ``On complex eigenvalues of {M} and {P} matrices,'' {\em Numerische
  Mathematik}, vol.~19, no.~2, pp.~170--175, 1972.

\bibitem{MathiasPang}
R.~Mathias and J.-S. Pang, ``Error bounds for the linear complementarity
  problem with a {P}-matrix,'' {\em Linear Algebra and Its Applications},
  vol.~132, pp.~123--136, 1990.

\end{thebibliography}
%%%%%%%%%%%%
\end{document}